\documentclass[11pt,twoside]{article}
\usepackage[margin=1in]{geometry}

\usepackage[utf8]{inputenc} % allow utf-8 input
\usepackage[T1]{fontenc}    % use 8-bit T1 fonts
\usepackage{hyperref}       % hyperlinks
\usepackage{url}            % simple URL typesetting
\usepackage{booktabs}       % professional-quality tables
\usepackage{amsfonts}       % blackboard math symbols
\usepackage{nicefrac}       % compact symbols for 1/2, etc.
\usepackage{microtype}      % microtypography
\usepackage{tikz}
\definecolor{light-gray}{gray}{0.9}

\usepackage{graphicx}
\usepackage{booktabs} % for professional tables
\usepackage{authblk}
\usepackage{amsmath,amssymb,amsthm,mathtools}
\usepackage{verbatim,makeidx}
\usepackage{algorithm}
\usepackage{algorithmic}
\usepackage{libertine,multirow}
\usepackage[libertine]{newtxmath}
\usepackage{color}
\usepackage{colortbl}
\usepackage{dsfont}
\usepackage{caption}
\usepackage{subcaption}
\usepackage{graphicx}
\usepackage{verbatim}
\usepackage{url,breakurl,hyperref}
\usepackage{pbox}
\usepackage{authblk}
\usepackage{bm}
\usepackage{enumitem}
\usepackage{comment}

\newcounter{ALC@tempcntr}% Temporary counter for storage

\theoremstyle{plain}% default
\newtheorem{theorem}{Theorem}%[section]
\newtheorem{proposition}{Proposition}%[section]
\newtheorem{lemma}{Lemma}%[section]
\newcommand{\beq}{\begin{eqnarray}}
\newcommand{\eeq}{\end{eqnarray}}

\newfont{\bbb}{msbm10 scaled 500}

\newfont{\bb}{msbm10 scaled 1100}

\newcommand{\av}{{\bf a}}

\newcommand{\xv}{{\bf x}}

\newcommand{\Ec}{{\cal E}}

\newcommand{\Hc}{{\cal H}}

\newcommand{\Pc}{{\cal P}}

\DeclareMathOperator{\Var}{var}

\newcommand{\remove}[1]{}

\newcommand{\avg}{{\mathbb E}}

\newcommand\reals{{\mathbb R}}

\newcommand\ff{{\mathbb F}}

\theoremstyle{definition}

\theoremstyle{remark}

\newcommand{\latexe}{{\LaTeX\kern.125em2%
                      \lower.5ex\hbox{$\varepsilon$}}}

\chardef\bslash=`\\	

\makeatletter		
\def\square{\RIfM@\bgroup\else$\bgroup\aftergroup$\fi
\vcenter{\hrule\hbox{\vrule\@height.6em\kern.6em\vrule}
\hrule}\egroup}\makeatother\makeindex

\DeclareMathAlphabet{\mathpzc}{OT1}{pzc}{m}{it}

\newcommand{\soumya}[1]{\textcolor{red}{#1}}
\newcommand{\arya}[1]{\textcolor{blue}{#1}}

\setlist[itemize]{leftmargin=0.1in}
\setlist[enumerate]{leftmargin=0.1in}

\title{High Dimensional Discrete Integration over the Hypergrid}

\author{Raj Kumar Maity}
\author{Arya Mazumdar}
\author{Soumyabrata Pal}
\affil{College of Information and Computer Sciences\\ University of Massachusetts Amherst   Amherst, MA 01003, USA\\ E-mail: rajkmaity@cs.umass.edu,  arya@cs.umass.edu, soumyabratap@umass.edu.}

\begin{document}

\maketitle

%%%%%%%%%%%%%%%%%%%%%%%%%%%%%%%%%%%%%%%%%%%%%%%%%%%%%%%%
%%%%%%%%%%%%%%%%%%%%%%%%%%%%%%%%%%%%%%
% ABSTRACT
%%%%%%%%%%%%%%%%%%%%%%%%%%%%%%%%%%%%%%
%%%%%%%%%%%%%%%%%%%%%%%%%%%%%%%%%%%%%%%%%%%%%%%%%%%%%%%%

\begin{abstract}
Recently Ermon et al. (2013) pioneered a way to practically compute approximations to large scale counting or discrete integration problems by using random hashes. The hashes are used to reduce the counting problem into many separate discrete optimization problems. The optimization problems then can be solved by an NP-oracle such as commercial SAT solvers or integer linear programming (ILP) solvers. In particular, Ermon et al. showed that if the domain of integration is  $\{0,1\}^n$ then it is possible to obtain a solution within  a factor of $16$ of the optimal (a 16-approximation) by this technique. 

In many crucial counting tasks, such as computation of partition function of ferromagnetic Potts model, the domain of integration is naturally $\{0,1,\dots, q-1\}^n, q>2$, the hypergrid. The straightforward extension of Ermon et al.'s method allows a $q^2$-approximation for this problem. For large values of $q$, this is undesirable. In this paper, we show an improved technique to obtain an approximation factor of $4+O(1/q^2)$ to this problem. We are able to achieve this by using an idea of optimization over multiple bins of the hash functions, that can be easily implemented by inequality constraints, or even in unconstrained way. Also the burden on the NP-oracle is not increased by our method (an ILP solver can still be used). %Our method extends to the case when the domain of integration is the symmetric group, and as a result we can obtain a $(4+o(1))$-approximation of the {\em permanent of} a matrix. 
%All these results hold assuming the existence of an NP-oracle. 
We provide experimental simulation results to support the theoretical guarantees of our algorithms. %, including comparison to the popular Markov-Chain-Monte-Carlo (MCMC) methods. 
\end{abstract}
%For example, this method only leads to a $n^2$
%, and if they allow some amenable  structure

 %A  of Ermon et al.'s work would allow a $q^2$-approximation for this problem, assuming the existence of an optimization oracle.
%  In this paper, we show that it is possible to obtain a $(2+\frac2{q-1})^2$-approximation to the discrete integration problem, when $q$ is a power of an odd prime (a similar expression follows for even $q$). We are able to achieve this by using an idea of optimization over multiple bins of the hash functions, that can be easily implemented by inequality constraints, or even in unconstrained way. Also the burden on the NP-oracle is not increased by our method (an LP solver can still be used). Furthermore, we provide a close-to-4-approximation for the permanent of a matrix by extending our technique. Note that, the domain of integration here is the symmetric group. Finally, we provide memory optimal hash functions that uses minimal number of random bits for the above purpose. We are able to obtain these structured hashes without sacrificing the amenability of the NP-oracle.  %was : therefore we need to come up with new structured hash families, that also turned out to be implementation-friendly. 

\section{Introduction}
Large scale counting problems, such as computing the permanent of a matrix or computing the partition function of a graphical probabilistic generative model, come up often in variety of inference tasks. These problems can, without loss of any generality, be written as {\em discrete integration}:  the summation of evaluations of a nonnegative function $w:\Omega \to \reals_+ \cup \{0\}$ over all elements of $\Omega$:
\begin{align}\label{eq:main}
S_\Omega(w) \equiv \sum_{\sigma \in \Omega} w(\sigma).
\end{align}
%The exponential (and sometime super-exponential) size of the domain makes these problems intractable. 
These problems are computationally intractable because of the exponential (and sometime super-exponential) size of $\Omega$. A special case is the set of problems \#P, counting problems associated with the decision problems in NP. For example, one might ask how many variable assignments a given CNF (conjunctive normal form) formula satisfies. The complexity class \#P was defined by Valiant \cite{valiant1979complexity}, in the context of computing the {\em permanent} of a matrix. The permanent of a matrix $A$ is defined as,
\begin{align}\label{eq:perm}
{\rm Perm}(A) \equiv \sum_{\sigma \in S_n} \prod_{i=1}^n A_{i,\sigma(i)},
\end{align}
where $S_n$ is the symmetric group of $n$ elements and $A_{i,j}$ is the $(i,j)$-th element of $A$. Clearly, here $S_n$ is playing the role of $\Omega$, and   $w(\sigma) = \prod_{i=1}^n A_{i,\sigma(i)}$. Therefore computing permanent of a nonnegative matrix is a canonical example of a problem defined by eq.~\eqref{eq:main}. 

Similar counting problems arise when one wants to compute the partition functions of the well-known probabilistic generative models of statistical physics, such as the Ising model, or more generally the Ferromagnetic Potts Model \cite{potts1952some}. Given a graph $G(V,E)$, and a label-space $Q \equiv \{0,1,2,\dots, q-1\}$, the partition function  $Z(G)$ of the Potts model is given by,
\begin{align}\label{eq:part}
 \sum_{\sigma \in Q^{|V|}}& \exp\Big(- \zeta\Big(J\sum_{(u,v) \in E} \delta(\sigma(u),\sigma(v)) 
+ H\sum_{u \in V} \delta(\sigma(u),0)  \Big) \Big),
\end{align}
where $\zeta$, $J$ and $H$ are system-constants (representing the temperature, spin-coupling and external force),   $\delta(x,y)$ is the delta-function that is $1$ if and only if $x=y$ and otherwise $0$, and $\sigma$ represents a label-vector, where $\sigma(u)$ is the label of vertex $u$.

It has been shown that, under the availability of an NP-oracle, every problem in \#P can be approximated within a factor of $(1\pm \epsilon), \epsilon >0$, with high probability via a randomized algorithm \cite{stockmeyer1985approximation}. This result says \#P can be approximated by ${\rm BPP}^{\rm NP}$ and the power of an NP-oracle and randomization is sufficient.  However, depending on the weight function $w(\cdot)$,  eq.~\eqref{eq:main} may not be in \#P. There are related approaches to count the number of
models of propositional formulas based on SAT-solvers, such as \cite{birnbaum1999good,bayardo2000counting,wei2005new,pesant2005counting,chakraborty2014distribution,chakraborty2016algorithmic} among others. 

The standard techniques to evaluate eq.~\eqref{eq:main} include the very influential fast variational methods \cite{wainwright2008graphical}, and Markov-Chain-Monte-Carlo based sampling schemes \cite{jerrum1996markov}. In practice, except for limited number of cases, these approaches are mostly used in a heuristic manner without nonasymptotic qualitative guarantees. 
Recently, Ermon et al. proposed an alternative approach (that they call \texttt{WISH} - Weighted-Integrals-And-Sums-By-Hashing) to solve these counting problems \cite{ermon2013taming,ermon2013optimization} by breaking them into multiple optimization problems. Namely, they  use families of hash functions $h: \Omega \to \tilde{\Omega}, |\tilde{\Omega}| < |\Omega|,$ and use a (possibly NP) oracle that can return the correct solution of the  optimization problem:
$
\max_{\sigma: h(\sigma) = a} w(\sigma),
$
for any $a \in \tilde{\Omega}$.
We call this oracle a MAX-oracle. In particular, when $\Omega=\{0,1\}^n$, and $h(\cdot)$ is a random hash function, assuming the availability of a MAX-oracle, Ermon et al. \cite{ermon2013taming} propose a randomized algorithm that approximates the discrete sum within a factor of  sixteen (a 16-approximation) with high probability.
Ermon et al. use simple linear sketches over $\ff_2$ (the finite field of size 2), i.e., the hash function $h_{A,b}:\ff_2^n \to \ff_2^m, A \in \ff_2^{m \times n}, b \in \ff_2^m$ is defined to be
\begin{align}\label{eq:hash}
h_{A,b}(x) = Ax + b,
\end{align}
where the arithmetic operations are over $\ff_2$. The matrix $A$ and the vector $b$ are randomly and uniformly chosen from the respective sample spaces. The MAX-oracle in this case simply provides solutions to the optimization problem:
$
\max_{\sigma \in \ff_2^n: A\sigma = b} w(\sigma).
$

The constraint space $\{\sigma \in \ff_2^n: A\sigma = b\}$ is nice since it is a coset of the nullspace of $A$, and experimental results showed them to be manageable by  optimization softwares/SAT solvers. In particular it was observed that being Integer Programming constraints, real-world instances are often solved in reasonable time. Since the implementation
of the hash function heavily affects the runtime, it makes sense to keep constraints of the MAX-oracle as an affine space as above. These constraints are also called parity constraints. The idea of using such constraints to show reduction among class of problems appeared in several papers before, including \cite{sipser1983complexity,valiant1986np,gomes2006model,thurley2011approximation,gomes2007counting} among others. 
%and was used famously by Valiant and Vazirani \cite{valiant1986np} where parity constraints were used to reduce SAT to Unique SAT. The idea has since been applied to various counting problems including \#SAT \cite{gomes2006model}, \#k-SAT \cite{thurley2011approximation}, and \#CSP \cite{gomes2007counting}.
The key property that the hash functions $\{h_{A,b}\}$ satisfy is that they are pairwise independent. This property can be relaxed somewhat - and in a subsequent paper Ermon et al. show that a hash family would work even if the matrix $A$ is sparse and random, thus effectively reducing the randomness as well as making the problem more tractable empirically \cite{ermon2014low}. Subsequently, Achlioptas and Jiang \cite{achlioptas2015stochastic} have shown another way of achieving similar guarantees. Instead of arriving at the set $\{\sigma \in \ff_2^n: A\sigma = b\}$ as a solution of a system of linear equations (over $\ff_2$), they view the set as the image of a lower-dimensional space. This is akin to the generator matrix view of a linear error-correcting code as opposed to the parity-check matrix view. This viewpoint allows their MAX-oracle  to solve just an unconstrained optimization problem.

\vspace{0.02in}
\noindent{\bf Drawbacks of obvious extensions of \cite{ermon2013taming} to large alphabets}.
 Note that, some crucial counting problems, such as computing the partition function of the Ferromagnetic Potts model of Eq.~\eqref{eq:part}, naturally have $\Omega= \{0,1, \dots, q-1\}^n, q>2$, i.e., a hypergrid. It is worth noting that while there exists polynomial time approximation (FPRAS) for the Ising model ($q=2$),  FPRAS for general Potts model ($q >2$) is significantly more challenging (and likely impossible \cite{goldberg2012approximating}). There are a few possible obvious extensions of Ermon et al. \cite{ermon2013taming} to larger alphabets.
 \begin{itemize}
 \item  (The straightforward extension).  The method  of \cite{ermon2013taming} can be used for $q$-ary in stead of binary. However, the drawback is that it  provides a $q^2$-approximation at best which is  particularly bad if $q$ is large (or growing with $n$). %By using a procedure outlined in \cite{ermon2013taming},  
 
 \item (Convert $q$-ary to binary).
To use the binary-domain algorithm of \cite{ermon2013taming} for any $\Omega =\{0,1, \dots, q-1\}^n$, we need to use a look-up table to map $q$-ary numbers to binary. In this process the number of variables (and also the number of constraints) increases by a factor of $\log q$. This makes the MAX-oracle significantly slower, especially when $q$ is large. Also, for the permanent problem, where $|\Omega|=\exp(n\log n)$,  this  creates a computational bottleneck. It would be useful to extend the method of   \cite{ermon2013taming} for  $\Omega= \ff_q^n$ without increasing the number of variables.

Furthermore, when $q$ is not a power of $2$, by converting $q$-ary configurations to binary, we introduce exponentially many invalid configurations. To account for these, the MAX-oracle must be adjusted accordingly which is a difficult task. This motivates us to keep the problem in its original domain and not convert the domain to binary. 

\item For the binary setting, it has been noted in \cite[section 5.3]{ermon2013taming} that the approximation ratio can be improved to any  $\alpha>1$ by increasing the number of variables, which extends to this $q$-ary setting. However this also results in an increase in number of variables by a factor of $\log_\alpha (q^2)$ which is undesirable.
\end{itemize}

\vspace{0.02in}
\noindent{\bf Our contributions.}
Our first contribution in this paper is to provide a new and improved algorithm to handle counting problems over nonbinary domains.
%new algorithm by modifying and generalizing the algorithm of \cite{ermon2013taming}. 
%nontrivial generalization of the algorithm of \cite{ermon2013taming}. 
For any hypergrid $\Omega = \{0,1, \dots, q-1\}^n, q$ is a power of prime, our algorithm provides a $4(1+\frac{1}{q-1})^2$-approximation, when $q$ is odd, and $4(1+\frac{2}{q-2})^2$-approximation, when $q>2$ is even, to the optimization problem of \eqref{eq:main} assuming availability of the MAX-oracle. 
Our algorithm utilizes an idea of using optimization over multiple bins of the hash function that can be easily implemented via inequality constraints.  The constraint space of the MAX-oracle remains an affine space and still can be represented as a modular integer linear program (ILP). 
% where our setting would be a fit.
%In general, for arbitrary $\Omega$, if represented as  $\{0,1\}^n$, the approximation factor is at best $16$ by the  technique of \cite{ermon2013taming}. But by having it represented as $\{0,1,\dots, q-1\}^n$  the approximation factor can be improved to $\sim4$ by our technique. 
Our multi-bin technique can also be used to extend the generator-matrix based algorithm of Achlioptas and Jiang \cite{achlioptas2015stochastic}. As a result, we need the MAX-oracle to only perform unconstrained maximization, as opposed to constrained. This lead to significant speed-up in the system, while resulting in the same approximation guarantees.

%Since we are delegating the hard task to a commercial optimization solver, our method can still be of interest here.  

Finally, we show the performance of our algorithms to compute the partition function of the ferromagnetic Potts model by running experiments on both synthetic datasets and real-worlds datasets. While in this paper we concentrate on theoretical results, the experiments serve as good `proof of concepts' for applications. We also use our algorithm to compute the Total Variation (TV) distance between two joint probability distributions over a large number of variables. In addition to comparing with the straightforward generalization of Ermon et al.'s method \cite{ermon2013taming}, we also show comparisons with  the popular Markov-Chain-Monte-Carlo (MCMC) method and the belief propagation method for discrete integration. All the experiments exhibit  good performance guarantees.

\bigskip
\noindent{\bf Organization.}
The paper is organized as follows. In Section~\ref{sec:back} we describe the technique by \cite{ermon2013taming} called the \texttt{WISH} algorithm, and then elaborate  our new ideas and  main results. In Section~\ref{sec:wish}, we provide the main technical results that lead to an improved approximation. We provide an algorithm with unconstrained optimization oracle (similar to  \cite{achlioptas2015stochastic}) and its analysis in Section~\ref{sec:mba}.  %In Section~\ref{sec:der}, we show how to optimally derandomize the hash function used in our algorithm.  
The experimental results on computation of partition functions and total variation distance are provided in Section~\ref{sec:exp}. %Some of the proofs and some experimental results are delegated to the appendix. % in the supplementary material.

While only of auxiliary interest here, we note that it is possible to derandomize the hash families based on parity-constraints to the optimal extent while maintaining the essential properties necessary for their performance. Namely, it can be ensured that the hash family can still be represented as $\{x \mapsto Ax+b\}$ while using information theoretically optimal memory to generate them. We discuss this in Appendix~\ref{sec:der}.

It turns out that, by using our technique and some modifications to the MAX-oracle, it is possible to obtain close-to-$4$-approximation to the problem of computing permanent of nonnegative matrices (assuming existence of NP-oracles). The NP-oracle still is amenable to be implemented in a commercial optimization solver. The idea of optimization over multiple bins is crucial here, since the straightforward generalization of Ermon et al.'s result would have given an approximation factor of $\Omega(n^2)$. 
Since there exists polynomial time randomized approximation scheme ($1\pm \epsilon$-approximation) of permanent of a nonnegative matrix \cite{jerrum2004polynomial}, the point of this exercise is to show that our method extends to find permanent of a matrix (albeit not with the best guarantees). We discuss this in Appendix~\ref{sec:perm}. % is devoted  to computation of permanent of a matrix. 

%While there already exists a polynomial time randomized approximation scheme ($1\pm \epsilon$-approximation) of permanent of a nonnegative matrix \cite{jerrum2004polynomial}, the runtime there is $\tilde{O}(n^{10})$. We note that the point of this exercise is to show that our method can be used to find permanent of a matrix (albeit not with the best guarantees).

\section{Background and our techniques}\label{sec:back}
In this section we describe the main ideas developed by \cite{ermon2013taming} and provide an overview of the techniques that we use to arrive at our new results.

%First of all assume without loss of generality that the configurations in $\Omega$ are ordered according to the value of the function $w$.

Let the elements in $\Omega$ be $\sigma_1, \sigma_2, \dots, \sigma_{|\Omega|}$ arranged according to a decreasing order of their weight, i.e., $w(\sigma_1) \ge w(\sigma_2) \ge   \dots \ge w(\sigma_{|\Omega|}).$ Let $\beta_i = w(\sigma_{q^i})$, for $i =0,1, \dots, n'$, where $n'$ is the smallest integer such that $q^{n'} \ge |\Omega|$. When $q^{n'}>|\Omega|$ we set $\beta_{n'} =0$ .

%\remove{
%First of all, notice that from \eqref{eq:main} we obtain:
%$
%S_\Omega(w) = \sum_{u\ge 0}u |\{\sigma \in \Omega: w(\sigma) =u \}|
% = \sum_{u\ge 0} |\{\sigma \in \Omega: w(\sigma) \ge u \}| = \sum_{u\ge 0} T(u),
%$
%where $T(u) \equiv |\{\sigma \in \Omega: w(\sigma) \ge u \}| $ is the tail distribution of weights and a nonincreasing function of $u$. Note that, $0 \le T(u) \le |\Omega|$. We can split the range of $T(u)$ into geometrically growing values $1, q, q^2, \dots, q^{n'}$ where $n'$ is the smallest integer such that $q^{n'} \ge |\Omega|$.
%%such that $q^{n'}\ge |\Omega|$. Let $\beta_i =u: T(u) = q^i, i=0, 1, \dots, n'$. \\\\
%%\soumya{The above definition is not correct since $T(u)$ can be discontinuous as reviewer 1 pointed out. The correct sentence should be the following.
%Let $\beta_i,$ $\forall i:q^i \le |\Omega|$  be the largest number such that $T(\beta_i) \ge q^i$ and $\beta_{i}=0$,$\forall i:q^i > |\Omega|.$ %Analogously, in the case when $q^{n'} \ge |\Omega|$, we can also assume the existence of $q^{n'}-|\Omega|$ extra configurations whose weights are zero. It is clear that this assumption does not change the value of $S_{\Omega}(w)$.
%%(Probably we need to show the picture in Ermon's paper to make this clearer.) 

Clearly $\beta_0 \ge \beta_1 \ge \dots \ge \beta_{n'}$. As we have not made any assumption on the values of the weight function, $\beta_i$ and $\beta_{i+1}$ can be far from each other.
% and they are hard to bound despite the fact that $T(u)$  is monotonic in nature. 
On the other hand we can try to bound the sum $S_\Omega(w)$ by bounding the area of the slice between $\beta_i$ and $\beta_{i+1}$. This area is at least $q^i(\beta_i-\beta_{i+1})$ and  at most $q^{i+1}(\beta_i-\beta_{i+1})$. Therefore:
$
\sum_{i=0}^{n'-1} q^{i}(\beta_{i}-\beta_{i+1})+ q^{n'}\beta_{n'}  \le S_\Omega(w) \le \sum_{i=0}^{n'-1} q^{i+1}(\beta_{i}-\beta_{i+1})+ q^{n^{'}}\beta_{n'}$ which implies
\vspace{-0.1in}
\begin{align} \beta_0 +(q-1) \sum_{i=1}^{n'}q^{i-1}\beta_{i}  &\le S_\Omega(w) \le  \beta_0 +(q-1) \sum_{i=1}^{n'}q^{i}\beta_{i}.\label{eq:bound}
\end{align}
%We should note here that $\beta_{n'}=0$ when $q^{n'}> |\Omega|$.

Hence $\beta_0 +(q-1) \sum_{i=1}^{n'}q^{i-1}\beta_{i}$ is a $q$-factor approximation of $S_\Omega(w)$ and if we are able to find a $k$-approximation of each value of $\beta_{i}$ we will be able to obtain a $kq$-factor approximation of $S_\Omega(w)$. In \cite{ermon2013taming}, subsequently the main idea is to estimate the coefficients $\{\beta_i, 0\le i \le n'\}$.

Now note that,
$
q^i = |\{\sigma \in \Omega: w(\sigma) \ge \beta_i \}|,
$
for $i =0,1, \dots, n'-1$. This also hold for $i =n'$ unless $q^{n'}> |\Omega|$ in which case $\beta_{n'} =0$.
Suppose, using a random hash function $h:\Omega \to \{0,1, \dots, q^i-1\}$ we compute hashes of all elements in $\Omega$. The pre-image of an entry in $\{0,1, \dots, q^i-1\}$ is called the {\em bin} corresponding to that value, i.e., $\{\sigma \in \Omega: h(\sigma) = x\}$ is the bin corresponding to the value $x\in \{0,1, \dots, q^i-1\}$. In every  bin for the hash function, there is on average   one element $\sigma$ such that $w(\sigma) \ge \beta_i$.  So for a randomly and  arbitrarily chosen bin $x \in  \{0,1, \dots, q^i-1\}$, if $w^\ast = \max_{\sigma: h(\sigma) = x} w(\sigma)$, then $w^\ast$ is a `good' approximation of  $\beta_i$ (this will be made rigorous later).
%\soumya{Reviewer 1 says that he does not follow but I do not see how to explain anymore.} 
 Indeed, suppose one performs this random hashing $\ell = O(\log n')$ times and then take the aggregate (in this case the median) value of $w^\ast$s. That is say, $\hat{w^\ast} = {\rm median}(w^\ast_1, \dots, w^\ast_\ell)$.
Then by using the independence of the hash functions, it can be shown that the aggregate is an upper bound on $\beta_i$ with high probability. 
%Indeed, without loss of generality,  if we assume that the configurations within $\Omega$ are ordered according to the value of the function $w$, i.e., $w(\sigma_1) \ge w(\sigma_2) \ge \dots \ge w(\sigma_{|\Omega|})$ then we can take $\beta_i = w(\sigma_{q^i})$.
In \cite{ermon2013taming},  $\Omega=\ff_2^n$ and  if the hash family is pairwise independent,  then by using the Chebyshev inequality it was shown that
$\hat{w^\ast} \in [\beta_{i+2}, \beta_{i-2}]$ with high probability. The \texttt{WISH} algorithm proposed by  \cite{ermon2013taming} makes use of the above analysis and provides a $2^{2\cdot2}=16$-approximation of $S_w(\Omega)$.  If we naively extend this algorithm for $S_w(\Omega)=\ff_q^n, q>2,$ then 
 it can be shown that
$\hat{w^\ast} \in [\beta_{i+1}, \beta_{i-1}]$ with high probability. This results in an
 approximation factor of $q^{2\cdot 1}=q^2$.
For example, for a ternary alphabet, $\Omega=\ff_3^n$,  we have a $9$-approximation to $S_w(\Omega).$

% This lead to a $q^{2c}$-approximation for $S_\Omega(w)$. For $c=2$ this leads to the $16$-approximation, because  and took $q=2$. 
 %Note that, for $q=2$, it was not possible to take $c=1$, but as we will see later that it is possible to take $c=1$ when $q>2$, and  for $q=3$, this observation immediately gives a $9$-approximation to $S_w(\Omega)$.

Instead of using a straightforward  analysis for the $q$-ary case, in this paper we use a MAX-oracle that can optimize over multiple bins of the hash function. Using this oracle we proposed a modified \texttt{WISH} algorithm and call it \texttt{MB-WISH} (Multi-Bin \texttt{WISH}). Just as in the case of  \cite{ermon2013taming,ermon2014low}, the MAX-oracle constraints can be integer linear programming constraints and commercial softwares such as CPLEX can be used.

The main intuition of using an optimization over multiple bins is that it boosts the probability that the $w^\ast$ we are getting above is close to $\beta_i$. % In particular, if we choose $r^m$ bins, then there are on average $r^m$ elements $\sigma$ in the bins such that $w(\sigma)\ge \beta_i$. So the multi-bin optimization is going to be an upper-estimate of $\beta_i$ with much higher probability.  
To be precise, we redefined $\beta_i \equiv w(\sigma_{(\frac{q}{r})^i })$ or   $(\frac{q}{r})^i = |\{\sigma \in \Omega: w(\sigma) \ge \beta_i \}|,$ for $i =1, 2, \dots, n' \equiv\lceil n \log_{q/r} q\rceil$. If we define $T(u) \equiv |\{\sigma \in \Omega: w(\sigma) \ge u \}| $, then Figure \ref{fig:multi-bin} illustrate the $T(u)$ vs. $u$ curve and locates $\beta_i$s therein. Note that, we would like to find the area under the $T(u)$ vs. $u$ curve, for which we use the sum of the vertical slices. Now to estimate the new $\beta_i$, we choose a hash function as before, and optimize over $r^i$ bins of the hash function. These steps are made rigorous in Section~\ref{sec:wish}. %the proof of 
However if we restrict ourselves to the binary alphabet then (as will be clear later) there is no immediate way to represent such multiple bins in a compact way in the MAX-oracle. For the non-binary case, it is possible to represent multiple bins of the hash function as simple inequality constraints.

% \soumya{The $\beta_i$'s for the non-binary and multi-bin setting is illustrated in Figure \ref{fig:multi-bin}.}
{\em This idea leads to an improvement in the approximation factor of $S_w(\Omega)$ to $4+\epsilon$, where $\epsilon$ decays to $0$ proportional to $q^{-1}$. Note that we need to choose $q$ to be a power of prime so that $\ff_q$ is a field.}

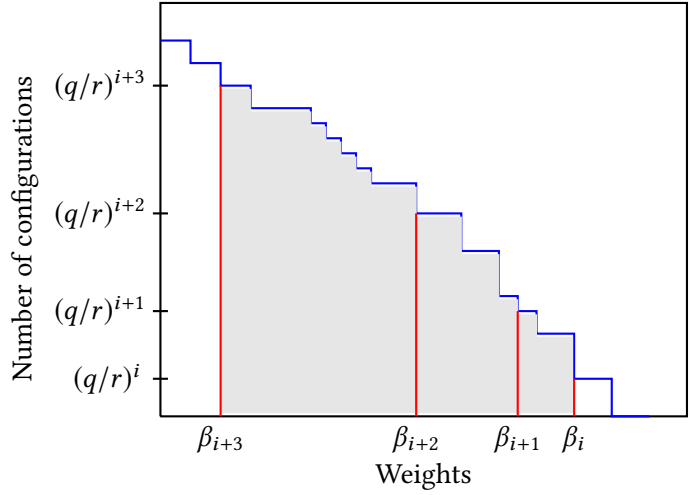
\begin{figure}
\centering
\begin{tikzpicture}

\draw [thick](0,0) -- (7,0) -- (7,5.5) -- (0,5.5) -- (0,0);
\draw [thick,blue] (0,5) -- (.4,5) -- 
    (.4,4.7) -- (.8,4.7) -- 
    (.8,4.4) -- (1.2,4.4) -- 
    (1.2,4.1) -- (2,4.1)  --
    (2, 3.9) -- (2.2, 3.9) --
    (2.2,3.7) -- (2.4,3.7)--
    (2.4,3.5) -- (2.6,3.5) --
    (2.6, 3.3) -- (2.8 , 3.3) --
    (2.8,3.1) -- (3.4,3.1)  --
    (3.4,2.7 ) -- (4, 2.7) --
    (4, 2.2) -- (4.5,2.2) --
    (4.5,1.6) -- (4.75,1.6) --
    (4.75,1.4) -- (5,1.4)--
    (5,1.1) -- (5.5,1.1) --
    (5.5,.5) -- (6,.5) --
    (6,0) -- (6.5,0);
    \filldraw [light-gray] (5,0.05) -- (5.48,0.05) -- (5.48,1.05) -- (5,1.05);
    \filldraw [light-gray] (4.75,0.05) -- (5,0.05) -- (5,1.35) -- (4.75,1.35);
    \filldraw [light-gray] (4.5,0.05) -- (4.75,0.05) -- (4.75,1.55) -- (4.5,1.55);
    \filldraw [light-gray] (4,0.05) -- (4.5,0.05) -- (4.5,2.15) -- (4,2.15);
    \filldraw [light-gray] (3.4,0.05) -- (4,0.05) -- (4,2.65) -- (3.4,2.65);
    \filldraw [light-gray] (2.8,0.05) -- (3.4,0.05) -- (3.4,3.05) -- (2.8,3.05);
    \filldraw [light-gray] (2.6,0.05) -- (2.8,0.05) -- (2.8,3.25) -- (2.6,3.25);
    \filldraw [light-gray] (2.4,0.05) -- (2.6,0.05) -- (2.6,3.45) -- (2.4,3.45);
    \filldraw [light-gray] (2.2,0.05) -- (2.4,0.05) -- (2.4,3.65) -- (2.2,3.65);
    \filldraw [light-gray] (2,0.05) -- (2.2,0.05) -- (2.2,3.85) -- (2,3.85);
    \filldraw [light-gray] (1.2,0.05) -- (2,0.05) -- (2,4.05) -- (1.2,4.05);
   \filldraw [light-gray] (.8,0.05) -- (1.2,0.05) -- (1.2,4.35) -- (.8,4.35);
\draw [thick ,red ] (.8,4.4) -- (.8,0) ;
\node at (-.8,4.4) {$(q/r)^{i+3}$};
\draw [thick] (-.1,4.4) -- (.1,4.4);
\node  at (.8,-.3) {$\beta_{i+3}$};    
\draw [thick, red] (3.4,2.7) -- (3.4,0);
\node at (-.8,2.7) {$(q/r)^{i+2}$};
\draw [thick] (-.1,2.7) -- (.1,2.7);
\node at (3.4,-.3) {$\beta_{i+2}$}; 
\draw [thick, red] (4.75,1.4) -- (4.75,0);
\node at (-.8,1.4) {$(q/r)^{i+1}$};
\draw [thick] (-.1,1.4) -- (.1,1.4);
\node at (4.75,-.3) {$\beta_{i+1}$};
\draw [thick, red] (5.5,.5) -- (5.5,0); 
\node at (-.7,.5) {$(q/r)^{i}$};
\draw [thick] (-.1,.5) -- (.1,.5);
\node at (5.5,-.3) {$\beta_{i}$};
\node at (-1.8,2.5) [ rotate=90] {Number of configurations};
\node at (3.5,-.8) {Weights };
\end{tikzpicture}
\caption{The $T(u)$ vs. $u$ curve and the illustration of $\beta_i$s.} \label{fig:multi-bin}
\end{figure}

In \cite{achlioptas2015stochastic}, the bins (as described above) are produced as images of some function, and not as pre-images of hashes. Since we want the number of bins to  be $q^i$, this can be achieved by looking at images of $g: \ff_q^{n-i} \to \Omega$ where $|\{g(\sigma): \sigma \in \ff_q^{n-i} \}| = q^{n-i}.$ The rest of the analysis of \cite{achlioptas2015stochastic} is almost same as above. The benefit of this approach is that the MAX-oracle just has to solve an unconstrained optimization  here.
Implementing our multi-bin idea for this perspective of \cite{achlioptas2015stochastic} is not straight-forward as we can no longer use inequality constraints for this. However, as we show later, we found a way to combine bins here in a succinct way generalizing the design of $g$. As a result, we get the same approximation guarantee as in \texttt{MB-WISH}, with the oracle load heavily reduced (this algorithm, that we call \texttt{Unconstrained MB-WISH},   can be found in Section~\ref{sec:mba}). %the Appendix).

%Given a permutation $\sigma$ the first step is to give it a random shift by composition with a random permutation $\pi$. The obtained permutation $\pi \circ \sigma$ is then treated as a vector we take an inner product of this vector with another randomly chosen vector over some field. Overall the hash can be written as,
%$$
%h(\sigma) = \langle p, \pi \circ \sigma\rangle + b \mod q,
%$$
%where $b$ is another random number between $0$ and $q-1$. It turns out that, while not pairwise independent, this random hash function, %while acting on two different configurations from  $S_n$, 
%possesses a similar property. We then use our original idea of optimization over multiple bins to obtain a good approximation for permanent.
%It turns out that this hash family, while not linear over any field,  still leads to a MAX-oracle that is implementable in common optimization softwares quite easily.

\section{The \texttt{MB-WISH} algorithm and analysis} \label{sec:wish}

%\soumya{I am changing the "<" sign to "$\prec$" sign wherever applicable.}
Let us assume $\Omega = \ff_q^n$ where $q$ is a prime-power. Let us also fix an ordering among the elements of $\ff_q\equiv  \{\alpha_0, \alpha_1,\dots , \alpha_{q-1}\}$ and write $\alpha_0 \prec \alpha_1 \prec \dots \prec \alpha_{q-1}$. In this section, the symbol `$\prec$' just signifies a fixed  ordering and has no real meaning over the finite field. Extending this notation, for any two vectors $x, y \in \ff_q^m$, we will say $x \prec y$ if and only if the $i$th coordinates of $x$ and $y$, satisfy $x_i<y_i$ for all $i=1, \dots, m$. Below $\mathbf{1}$ denotes an all-one vector of a dimension that would be clear from context. Also, for any event $\Ec$ let $\mathds{1}[\Ec]$ denote the indicator  for the event $\Ec$.

The MAX-oracle for \texttt{MB-WISH} performs the following optimization, given $A \in \ff_q^{m \times n}; b,s \in \ff_q^m$:
\vspace{-0.04in}
\begin{align}
\max_{\sigma \in \ff_q^n : A\sigma +b \prec s} w(\sigma).
\end{align}

The modified \texttt{WISH} algorithm is presented as Algorithm \ref{algo:mbwish}. The main result of this section is below.
\begin{algorithm}                   % enter the algorithm environment
\caption{\texttt{MB-WISH} algorithm for $\Omega = \ff_q^n$ and a weight function $w$ \label{algo:mbwish}}          % give the algorithm a caption
%\label{alg1}                           % and a label for \ref{} commands later in the document
\begin{algorithmic}                 % enter the algorithmic environment
      \REQUIRE $r$, $\gamma= \frac{q}{3r}(\frac{1}{2}-\frac{r}{q})^{2}$, $\ell = \lceil \frac{1}{\gamma}\ln \frac{2n}{\delta} \rceil$, $n'=\lceil n \log_{q/r} q\rceil$
      \STATE $M_0 \equiv \max_{\sigma\in \ff_q^n} w(\sigma)$
       \FOR{$i\in\{1,2,\dots,n'\}$}
         \FOR{$k \in \{1,\dots,\ell\}$}
           \STATE Sample hash functions $h_i \equiv h_{A^{i},b^{i}}$ uniformly at random from $\Hc_{i,n}$ as defined in \eqref{eq:fam}% $h_{A,b}\{0,1,\dots,q-1\}^{n} \rightarrow \{0,1,\dots,q-1\}^{m}$
           \STATE $w_{i}^{(k)}=\max_{\sigma: A^{i}\sigma+b^{i} \prec \alpha_r \cdot \mathbf{1}} w(\sigma) $
        \ENDFOR
         \STATE $M_{i}={\rm Median}(w_{i}^{(1)},w_{i}^{(2)},\dots,w_{i}^{(\ell)})$
       \ENDFOR
       \STATE Return $M_{0}+(\frac{q}{r}-1)\sum_{i=0}^{n'-1}M_{i+1}\big(\frac{q}{r}\big)^{i} $
\end{algorithmic}
\end{algorithm}

%\soumya{I changed $r=$ to $r\le$ in the theorem statement.} 
\begin{theorem}\label{thm:mb}
Suppose $q>2$ is a prime power, $\Omega = \ff_q^n$ and  a positive integer $r \le \lfloor{\frac{q-1}{2}}\rfloor$. For any  $\delta> 0$, % and a positive constant $\gamma \equiv \frac{q}{3r}(\frac{1}{2}-\frac{r}{q})^{2}$, 
Algorithm \ref{algo:mbwish} makes $\Theta(n \log \frac{n}{\delta})$ calls to the MAX-oracle, and with probability  $\ge 1-\delta$ outputs a $(\frac{q}{r})^{2}$-approximation of $S_w(\Omega)$. % where $\frac{r}{q} \le \frac{1}{2}-\epsilon$ 
\end{theorem}

By setting $r = \lfloor{\frac{q-1}{2}}\rfloor$, our algorithm provides a $4(1+\frac{1}{q-1})^2$-approximation, when $q$ is odd, and $4(1+\frac{2}{q-2})^2$-approximation, when $q>2$ is even.

The constant in the big-O term in the number of calls to the oracle is a function of $q$ and $r$. In particular, when $r = \lfloor{\frac{q-1}{2}}\rfloor$ and $q$ odd, this constant varies as $q^2\log q$. We can tune the value of $r$ to reduce the number of calls to the oracle at the expense of the approximation factor.

The theorem will be proved by a series of lemmas.
The key trick that we are using is to ask the MAX-oracle to solve an optimization problem over not a single bin, but {\em multiple bins} of the hash function. 
This is going to boost the probability that our estimates of $\beta_i$s are good. 
In particular we will solve the optimization over $r^m$ bins of the hash function. The hash family is defined in the following way. We have $h_{A,b}:\ff^n \to \ff^m: x \mapsto Ax+ b$, the operations are over $\ff_q$.
Let 
\begin{align}\Hc_{m,n} = \{h_{A,b}: A \in \ff_q^{m \times n}, b \in \ff_q^m\}.\label{eq:fam}
\end{align} 
For readers familiar with coding theory, the basis behind our technique is simple. The set of configurations $\{\sigma \in \ff_q^n : A\sigma ={\bf 0}\}$ forms a linear code of dimension $n-m$. The bins of the hash function define the cosets of this linear code. We would like to chose $q^r$ cosets of a random linear code and the find the optimum value of $w$ over the configurations of these cosets as the MAX-oracle. 
%\soumya{Should we skip the coding theory analogies?} 
To choose a hash function uniformly and randomly from $\Hc$, we can just choose the entries of $A$ and $b$ uniformly at random from $\ff_q$ independently. 

%Motivated with the $9$ approximation scheme with the help of trits $(\{0,1,2\})$, we are trying to obtain a better approximation than $9$ by representing the elements from the universe by vectors where each element belongs to $F_{q}$ and $q$ is a prime or a power of a prime number. 
%\begin{proposition}
%Let $A \in \{0,1,\dots,q-1\}^{m \times n}$  and $b \in \{0,1,\dots,q-1\}^{m}$ . Then the hash family $\mathds{H}=\{h_{A,b} :\{0,1,\dots,q-1\}^{n} \rightarrow \{0,1,\dots,q-1\}^{m} \}$
%where $h_{A,b}^{m}(x)=Ax+b$ is an $\epsilon-SU$ hash family where $\epsilon=\frac{1}{q^{m}}$. These hash families are also pairwise independent as can be observed easily.
%\end{proposition}
%Let us assume that the number of configurations that we have is $|\Sigma|=q^{n} $. The main trick that we are going to use over here is consider multiple buckets at the same time. In order to do this we are asking to oracle to solve the following optimization problem with inequality constraints $$\max_{\sigma: a_{i}\sigma+b_{i} < r \; \forall i} w(\sigma) $$ where $ a_{i} $ is the $i$-th row of the matrix $A$ and $b_i$ the $i$-th element of $b$ , $r (<q)$ is an integer and all operations are modulo q. Hence we have the following modified version of the algorithm.

%\begin{lemma}
Note that, the hash family $\Hc_{m,n}$ as defined in \eqref{eq:fam} is uniform and pairwise independent.
%\end{lemma}
%The proof of this fact is quite straightforward and we omit it here. 
It  follows from the following more general result.

\begin{lemma}\label{lem:pair}
Let us define  $Z_{\sigma}$ to be the indicator random variable denoting $A\sigma+b \prec \alpha_r\cdot \mathbf{1}$ for some $r\in \{0, \dots, q-1\}$ and $A,b$ randomly and uniformly sampled from $\Hc_{m,n}$. Then $\Pr(Z_{\sigma}=1)=\big(\frac{r}{q}\big)^{m}$ and for any two configurations $\sigma_{1},\sigma_{2}\in \ff_q^n$ the random variables $Z_{\sigma_{1}}$ and $Z_{\sigma_{2}}$ are independent.
\end{lemma}

%The proof is delegated to Appendix~\ref{sec:proof} in the supplementary material.

\begin{proof} %[Proof of Lemma~\ref{lem:pair}]
Let $A_i$ denote the $i$th row of $i$ and $b_i$ denote the $i$th entry of $b$. Then 
$\mathds{1}[A\sigma+b \prec \alpha_r\cdot \mathbf{1}]=\bigwedge_{i=1}^{m}\mathds{1}[A_{i}\sigma+b_i \prec \alpha_r]$. 

For all configurations $\sigma \in \Omega, \forall i$, we must have
\begin{align*}
\Pr(A_i\sigma +b_i \prec \alpha_r) =\sum_{j=0}^{r-1} \Pr(A_i\sigma +b_i =\alpha_j)=\frac{r}{q}.
\end{align*}
%since for each value of $j$ the events are disjoint}

As  $A_i,b_i$  are independent $1\le i \le m$, we must have that $\Pr(A\sigma +b \prec \alpha_r\cdot \mathbf{1})=(\frac{r}{q})^m$. Now for any two configurations $\sigma_1 ,\sigma_2\in \ff_q^n$,
\begin{align*}
& \Pr(A_i\sigma_1+b_i \prec \alpha_r \wedge A_i\sigma_2+b_i \prec \alpha_r ) \\ 
&=\sum_{k=0}^{r-1}\sum_{j=0}^{r-1}\Pr(A_i\sigma_1+b_i =\alpha_k \wedge A_i\sigma_2+b_i = \alpha_j )\\ 
&=\sum_{k=0}^{r-1}\sum_{j=0}^{r-1}\Pr(A_i\sigma_1+b_i =\alpha_k)\\
& \qquad \cdot\Pr(A_i\sigma_2+b_i = \alpha_j|A_i\sigma_1+b_i =\alpha_k)\\
&=\sum_{k=0}^{r-1}\sum_{j=0}^{r-1}\Pr(A_i\sigma_1+b_i =\alpha_k)\Pr(A_i(\sigma_2-\sigma_1) = \alpha_j-\alpha_k)\\
&=r^2(1/q)(1/q)=(r/q)^2.
\end{align*}
As all the rows are independent, $\Pr(A\sigma_1 +b \prec \alpha_r\cdot \mathds{1} \wedge A\sigma_2 +b \prec \alpha_r\cdot \mathds{1})=(\frac{r}{q})^{2m}.$
\end{proof}

%Hence we have pair-wise independence for these events and now we move on to analysis part.
Fix an ordering of the configurations $(\sigma_i,1\leq i \leq q^n)$ such that $1\leq j \leq q^n, w(\sigma_j)\geq w(\sigma_{j+1})$. We can also interpolate the space of configuration to make it continuous by the following technique. For any positive real number $x= z+f$, where $z =\lfloor x\rfloor $ is the integer part and $f = x-z$ is the fractional part, define $w(\sigma_x) = w(\sigma_z)$.  For $i \in \{0,1,2,\ldots, n'\equiv \lceil n \log_{q/r} q\rceil\}$, define $\beta_i=w(\sigma_{t^i}) = w(\sigma_{\lfloor t^i\rfloor})$, where $t =\frac{q}{r}$. We take $w(\sigma_k) =0$ for $k > q^n$. See Figure~\ref{fig:multi-bin} for an illustration.

To prove Theorem \ref{thm:mb} we need the following crucial lemma as well.
\begin{lemma}\label{lem:cru}
Let $M_{i}={\rm Median}(w_{i}^{(1)},\dots,w_{i}^{(\ell)})$ be defined as in the Algorithm \ref{algo:mbwish}. Then for $\gamma = \frac{q}{3r}(\frac{1}{2}-\frac{r}{q})^{2}$, we have, %the following condition must hold
$\Pr \bigg(M_i \in [\beta_{\min(i+1,n')},\beta_{\max(i-1,0)}] \bigg) \ge 1-2\exp(-\gamma \ell).$
\end{lemma}

\begin{proof} %[Proof of Lemma~\ref{lem:cru}]
 Consider the set of $\lfloor t^j\rfloor$ heaviest configuration 
\begin{align*}
\Omega_j=\{\sigma_1,\ldots \sigma_{\lfloor t^j \rfloor} \} .
\end{align*}
Let $S_j(h_i) = |\{\sigma \in \Omega_j: A^i\sigma+b^i \prec \alpha_r\cdot\mathbf{1}\}|.$ Recall $h_i$ is sampled uniformly at random from $\Hc_{i,n}.$
By the uniformity property of the hash function, 
\begin{align*}
\avg S_j(h_i)=\avg \sum_{\sigma \in \Omega_j}\mathds{1}[h_i(\sigma) \prec \alpha_r\cdot \mathbf{1}]=\frac{\lfloor t^j\rfloor}{t^i}.
\end{align*}
For each configuration $\sigma$ let us denote the random variable $\bar{Z^i_{\sigma}}=\mathds{1}[\{h_i(\sigma) \prec \alpha_r\cdot \mathbf{1} \}]-\frac{1}{t^i}$. By our design $\avg \bar{Z^i_{\sigma}}=0.$ Note that, $S_j(h_i) - \avg S_j(h_i) = \sum_{\sigma\in \Omega_j}\bar{Z^i_{\sigma}}$.
Also, from Lemma \ref{lem:pair}, the random variables $\bar{Z^i_{\sigma}}$s are pairwise independent.
Therefore,
\begin{align*}
\Var S_j(h_i) =\Var(\sum_{\sigma\in \Omega_j}\bar{Z^i_{\sigma}})=\sum_{\sigma\in \Omega_j}\avg \bar{Z^i_{\sigma}}^2=\frac{\lfloor t^j\rfloor}{t^i}(1-\frac{1}{t^i}).
\end{align*}
Now, % 
for any $1 \le k \le \ell$,
\begin{align*}
\Pr(w_i^{(k)}\geq \beta_j)=\Pr(w_i^{(k)} \geq w(\sigma_{\lfloor t^j\rfloor}))= \Pr(S_j(h_i)\ge1)
 = 1- \Pr(S_j(h_i)\le 0).
\end{align*}
Let $j =i+1$. Then, using Chebyshev inequality,
\begin{align*}
\Pr(S_j(h_i)\le 0) = \Pr\Big(S_j(h_i) - \avg S_j(h_i) \le - \frac{\lfloor t^j\rfloor}{t^i}\Big)  % \le \Pr\Big(S_j(h_i) - \avg S_j(h_i) \le - \frac{ t^j}{t^i}\Big)\\
\le  \frac{\Var S_j(h_i)}{(\frac{\lfloor t^j\rfloor}{t^i})^2} \le \frac{t^{i}(1-1/t^i)}{\lfloor t^{j}\rfloor} < \frac{t^i-1}{t^{i+1}-1} \le \frac{1}{t} = \frac{r}{q}.
\end{align*}
%By doing the same analysis as before we can easily show that for $j=i+c$
%\begin{align*}
%Pr[w_i\geq b_j]&=Pr[w_i\geq w(\sigma_{t^j})]\geq Pr[S_j(h^i)>1]\\
%& \geq 1- \frac{C_k}{(t^c)^{k/2}}
%\end{align*}
Therefore, 
\begin{align*}
\Pr(w_i^{(k)}\geq \beta_{i+1}) \ge 1 - \frac{r}{q}.
\end{align*}

Also,
$
\Pr(w_i^{(k)}\leq \beta_{i-1}) = \Pr(w_i^{(k)} \leq w(\sigma_{\lfloor t^{i-1}\rfloor})) \ge \Pr(S_{i-1}(h_i) =0).
$
Notice that the last inequality is satisfied because $S_{i-1}(h_i) =0$ implies  $w_i^{(k)} \leq w(\sigma_{\lfloor t^{i-1}\rfloor})$.
Now, continuing the chain of inequalities, using Markov inequality, 
\begin{align*}
\Pr(w_i^{(k)}\leq \beta_{i-1})  \ge 1 - \Pr(S_{i-1}(h_i) \ge 1) \ge 1 - \avg S_{i-1}(h_i) 
= 1-\frac{\lfloor t^{i-1}\rfloor}{t^i}\ge 1-\frac1t = 1-\frac{r}{q}.
\end{align*}
%For $c=1$ and with pairwise independent has function $(k=2 ,C_k=1)$
%\begin{align*}
%Pr[w_i\geq b_{i+c}]  \geq 1- \frac{C_k}{(t^c)^{k/2}}=1-\frac{1}{t}=1- \frac{r}{q}
%\end{align*}
%For $j=i-c$ and the same way we can bound the other side 
%\begin{align*}
%Pr[w_i\leq b_{i-c}]  &=Pr[S_{i-c}(h^i)\leq 0]=1-Pr[S_{i-c}(h^i) > 1] \\
%& \geq 1- \frac{\mathbb{E}{S_{i-c}}}{1} \quad \text{(Markov Inequality)}\\
%&=1-\frac{1}{t}=1-\frac{r}{q} 
%\end{align*}

%To have a bound more than $t>2$ we need $r < \frac{q}{2} $. we can show that 
Recall that, $M_i = {\rm Median}(w_i^{(1)}, w_i^{(2)}, \dots, w_i^{(\ell)})$. Define, $X_i^{(k)}$ to be the indicator random variable of the event $\{w_i^{(k)} \le \beta_{i+1}\}$. Therefore, $\Pr(M_i \le \beta_{i+1})= \Pr(\sum_{k=1}^\ell X_i^{(k)} \ge \ell/2).$ On the other hand, note that, $\Pr(X_i^{(k)} =1) \le r/q$. We know from Chernoff bound that, if $X$ is a sum of iid $\{0,1\}$ random variables then $\Pr(X \ge \avg X(1+\delta)) \le \exp(-\avg X \delta^2/3)$.
Therefore,
\begin{align*}
\Pr(M_i \le \beta_{i+1})\le \exp\Big(-\frac{\ell q}{3r}\Big(\frac12-\frac{r}{q}\Big)^2 \Big).
\end{align*}
%Pr[M_i\leq b_{i-c}]\geq 1-\exp(-\alpha' T)
Similarly, 
\begin{align*}
\Pr(M_i \ge \beta_{i-1})\le \exp\Big(-\frac{\ell q}{3r}\Big(\frac12-\frac{r}{q}\Big)^2 \Big).
\end{align*}
This proves the lemma.
%where $\alpha'=2(1-\frac{r}{q}-1/2)^2=2(1/2-\frac{r}{q})^2$ which will produce the result 
%\begin{align*}
%Pr[b_{i+c}\leq M_i \leq b_{i-c}]& \geq 1-2\exp(\alpha'T)\\
%&=1-\exp(\alpha^*T)
%\end{align*}
%where $\alpha^{*}=\alpha'\ln 2$
\end{proof}

%The proof of this lemma is again delegated to Appendix~\ref{sec:proof} along with the full proof of Theorem \ref{thm:mb}.
%Here we provide a sketch of the remaining steps.
From Lemma \ref{lem:cru}, the output of the algorithm  lies in the range $[L',U']$ with probability at least $1-\delta$ where 
$
L'=\beta_0+(t-1)\sum_{i=0}^{n'-1}\beta_{\min\{i+2,n'\}}t^i$ and $U'=\beta_0+(t-1)\sum_{i=0}^{n'-1}\beta_it^i$. $L'$ and $U'$ are a factor of $t^2$ apart.
Now, following an argument similar to \eqref{eq:bound}, we can show $L' \le S_w(\Omega) \le U'.$

%\begin{corollary}
%As a corollary $\Pr \bigg[ \bigcap_{i=0}^{n} M_i \in [b_{\min(i+c,n)},b_{\max(i-c,0)}] \bigg] \ge 1-n\exp(-\alpha T) \ge 1-\delta$ for $T=\ln \frac{\frac{1}{\delta}}{\alpha}\ln n$ and $M_{0}=b_{0}$
%\end{corollary}
%
%\begin{lemma}
%The above algorithm provides a $(\frac{q}{r})^{2}$ approximation to the problem
%\end{lemma}
%\begin{proof}
%The algorithm outputs $M_0+(t-1)\sum_{i=0}^{n'-1}M_{i+1}t^i$ which lies in the range $[L',U']$ with probability at least $1-\delta$ where 
%$
%L'=\beta_0+(t-1)\sum_{i=0}^{n'-1}\beta_{\min\{i+2,n'\}}t^i$ and $U'=\beta_0+(t-1)\sum_{i=0}^{n'-1}\beta_it^i.$

%Now notice that, as $ \beta_0 \ge \beta_1$, we have %we can prove the $t^{2}-$ approximation
%$
%U' =\beta_0+(t-1)\sum_{i=0}^{n'-1}\beta_it^i  
%= \beta_0 + (t-1) (\beta_0 +\beta_1t)+\sum_{i=2}^{n'-1}\beta_it^i 
% \leq t^2\beta_0+ t^2.(t-1). \sum_{i=2}^{n'-1}\beta_it^{i-2}
%\leq t^2(\beta_0+(t-1)\sum_{i=0}^{n'-1}\beta_{\min\{i+2,n'\}}t^i)=t^2L'.
%$

%The only thing that remains to be proved is that $L' \le S_w(\Omega) \le U'.$ However that is true, by just following an argument similar to \eqref{eq:bound}. Indeed,
%$
%\sum_{i=0}^{n'-1}\beta_{i+1}(t^{i+1} -  t^{i}) \le 
%S_\Omega(w) \le\sum_{i=0}^{n'-1}\beta_{i}( t^{i+1} - t^{i} )$
%which implies $L' \le S_\Omega(w) \le U'.
%%& \le \beta_0(t-1) + \sum_{i=1}^{n'-1} 
%$
%\begin{align*}
%\beta_0 +(t-1) \sum_{i=1}^{n'}t^{i-1}\beta_{i}  \le S_\Omega(w) \le  \beta_0 +(t-1) \sum_{i=1}^{n'}t^{i}\beta_{i}.
%\end{align*}
Therefore Algorithm \ref{algo:mbwish} provides a $t^2$-approximation to $S_\Omega(w)$. 
%The total number of calls to the MAX-oracle is
%$n'\ell +1 = O(n\log(n/\delta))$.
%\end{proof}
Let us now give the full proof of  Theorem  \ref{thm:mb}.
\begin{proof}[Proof of Theorem \ref{thm:mb}]
From Lemma \ref{lem:cru}, we have, 
$\Pr \bigg[ \bigcap_{i=1}^{n} M_i \in [\beta_{\min(i+1,n')},\beta_{\max(i-1,0)}] \bigg] \ge 1-2n\exp(-\gamma \ell) = 1-\delta$
 for $\ell = \frac{1}{\gamma}\ln\frac{2n}{\delta}$ and by definition $M_{0}=\beta_{0}.$

%\begin{corollary}
%As a corollary $\Pr \bigg[ \bigcap_{i=0}^{n} M_i \in [b_{\min(i+c,n)},b_{\max(i-c,0)}] \bigg] \ge 1-n\exp(-\alpha T) \ge 1-\delta$ for $T=\ln \frac{\frac{1}{\delta}}{\alpha}\ln n$ and $M_{0}=b_{0}$
%\end{corollary}
%
%\begin{lemma}
%The above algorithm provides a $(\frac{q}{r})^{2}$ approximation to the problem
%\end{lemma}
%\begin{proof}
The algorithm outputs $M_0+(t-1)\sum_{i=0}^{n'-1}M_{i+1}t^i$ which lies in the range $[L',U']$ with probability at least $1-\delta$ where 
\begin{align*}
L'=\beta_0+(t-1)\sum_{i=0}^{n'-1}\beta_{\min\{i+2,n'\}}t^i  \quad \text{ and } \quad  U'=\beta_0+(t-1)\sum_{i=0}^{n'-1}\beta_it^i.\end{align*}

Now notice that, as $ \beta_0 \ge \beta_1$, we have %we can prove the $t^{2}-$ approximation
\begin{align*}
U' &=\beta_0+(t-1)\sum_{i=0}^{n'-1}\beta_it^i  
\\& = \beta_0 + (t-1) (\beta_0 +\beta_1t)+(t-1)\sum_{i=2}^{n'-1}\beta_it^i \\
& \leq t^2\beta_0+ t^2.(t-1). \sum_{i=2}^{n'-1}\beta_it^{i-2}\\
&\leq t^2(\beta_0+(t-1)\sum_{i=0}^{n'-1}\beta_{\min\{i+2,n'\}}t^i)=t^2L'.
\end{align*}

The only thing that remains to be proved is that $L' \le S_w(\Omega) \le U'.$ However that is true, by just following an argument similar to \eqref{eq:bound}. Indeed,
\begin{align*}
\sum_{i=0}^{n'-1}\beta_{i+1}(t^{i+1} -  t^{i}) \le 
S_\Omega(w) \le\sum_{i=0}^{n'-1}\beta_{i}( t^{i+1} - t^{i} )\end{align*}
which implies $L' \le S_\Omega(w) \le U'.$
%& \le \beta_0(t-1) + \sum_{i=1}^{n'-1} 

%\begin{align*}
%\beta_0 +(t-1) \sum_{i=1}^{n'}t^{i-1}\beta_{i}  \le S_\Omega(w) \le  \beta_0 +(t-1) \sum_{i=1}^{n'}t^{i}\beta_{i}.
%\end{align*}
Therefore Algorithm \ref{algo:mbwish} provides a $t^2$-approximation to $S_\Omega(w)$. The total number of calls to the MAX-oracle is
$n'\ell +1 = O(n\log(n/\delta))$.
\end{proof}

To exemplify this result, suppose $q=3$. In this case the algorithm provides a $9$-approximation. Later, in the experimental section, we have used a ferromagnetic Potts model with $q=5$.  \texttt{MB-WISH} provides a $\frac{25}{4}=6.25$-approximation in that case. Note that, for a $5$-ary Potts model, it is only natural to use our algorithm instead of converting it to binary in conjunction with the original algorithm of Ermon et al.

%\begin{remark}
Instead of pairwise independent hash families, if we employ $k$-wise independent families, it leads to a better decay probability of error. However it does not improve the approximation factor. 
%\end{remark}

\vspace{0.02in}
\noindent\textbf{Unconstrained optimization oracle.} We can modify and generalize the results of Achlioptas and Jiang \cite{achlioptas2015stochastic} to formulate a version of \texttt{MB-WISH} that can use unconstrained optimizers as the MAX-oracle. The MAX-oracle for this algorithm performs an unconstrained optimization of the  form:  $
\max_{\sigma \in B}w(A\sigma +b)$, given $A\in \ff_q^{m\times n}, b \in \ff_q^n$ and a set $B \subseteq \ff_q^m$.

The aim is to carefully design $B$ so that all the desirable statistical properties are satisfied.
This part is quite different from the hashing-based analysis  and not an immediate extension of \cite{achlioptas2015stochastic}. We provide the algorithm (\texttt{Unconstrained MB-WISH}) and its analysis in the next section. % Section~\ref{sec:mba}.

\section{\texttt{MB-WISH} with unconstrained optimization oracle}
\label{sec:mba}
In this section, we provide an algorithm that uses unconstrained optimizations for the oracle, as in the case of Achlioptas and Jiang \cite{achlioptas2015stochastic}. We call this algorithm \texttt{Unconstrained MB-WISH}.

Let us assume $\Omega = \ff_q^n$ where $q$ is a prime-power. As before,  let us also fix an ordering among the elements of $\ff_q\equiv  \{\alpha_0, \alpha_1,\dots , \alpha_{q-1}\}$ and write $\alpha_0\prec \alpha_1 \prec \dots \prec\alpha_{q-1}$. Recall that, here the symbol `$\prec$' just signifies a fixed  ordering and has no real meaning over the finite field. %For any event $\Ec$ let $\mathds{1}[\Ec]$ denote the indicator random variable for the event $\Ec$.

The MAX-oracle for \texttt{Unconstrained MB-WISH} performs an unconstrained optimization of the following form, given $A\in \ff_q^{m\times n}, b \in \ff_q^n$ and a set $B \subseteq \ff_q^m$:
\begin{align}
\label{eq:umb}
\max_{\sigma \in B}w(A\sigma +b).
\end{align}

The \texttt{Unconstrained MB-WISH} algorithm is presented as Algorithm~\ref{algo:mbawish}. The main result of this section is the following.
%\soumya{Again change $r=$ to $r\le$ in the theorem statement.}
\begin{theorem}\label{thm:mba}
Suppose $q>2$ is a power of a prime and a positive integer $r \le \lfloor{\frac{q-1}{2}}\rfloor$. Let $\Omega = \ff_q^n$. For any  $\delta> 0$, % and a positive constant $\gamma = \frac{q}{3r}(\frac{1}{2}-\frac{r}{q})^{2}$, 
Algorithm \ref{algo:mbawish} makes $\Theta(n \log \frac{n}{\delta})$ calls to the MAX-oracle (cf.~\eqref{eq:umb}), and with probability at least $1-\delta$ outputs a $(\frac{q}{r})^{2}$-approximation of $S_w(\Omega)$. % where $\frac{r}{q} \le \frac{1}{2}-\epsilon$ 
\end{theorem}

\begin{algorithm}                 % enter the algorithm environment
\caption{\texttt{Unconstrained MB-WISH} algorithm for $\Omega = \ff_q^n$ and a weight function $w$\label{algo:mbawish}}          % give the algorithm a caption
%\label{alg1}                           % and a label for \ref{} commands later in the document
\begin{algorithmic}                 % enter the algorithmic environment
      \REQUIRE $\ell \rightarrow \lceil \frac{1}{\gamma}\ln \frac{2n}{\delta} \rceil$, $r,n'=\lceil n \log_{q/r} q\rceil$
      \STATE $M_0 \equiv \max_{\sigma\in \ff_q^n} w(\sigma)$
       \FOR{$i\in\{1,2,\dots,n\}$}
         \FOR{$k \in \{1,\dots,\ell\}$}
           \STATE Sample a full rank matrix uniformly at random from the set of all full rank $n \times n$ matrices in  $\ff_{q}^{n\times n}$ and construct matrices $A$ and $R$ by taking the first $n-i$ columns and the last $i$ columns respectively.  Sample $b \in \ff_{q}^{n}$ uniformly at random % $h_{A,b}\{0,1,\dots,q-1\}^{n} \rightarrow \{0,1,\dots,q-1\}^{m}$
           \STATE $w_{i}^{(k)}=\max_{\substack{x \in \ff_{q}^{n-i} \\ y \in \{\alpha_0,\alpha_1,\dots,\alpha_{r-1}\}^{i}}} w(Ax+Ry+b)$
        \ENDFOR
         \STATE $M_{i}={\rm Median}(w_{i}^{(1)},w_{i}^{(2)},\dots,w_{i}^{(\ell)})$
       \ENDFOR
       \FOR{$i\in\{n+1,\dots,n'\}$}
         \FOR{$k \in \{1,\dots,\ell\}$}
           \STATE Sample full rank matrix $A \in \ff_{q}^{n \times n}, b \in \ff_{q}^{n}$ uniformly at random. Set $\mathcal{S}_{i}$ as defined in Equation \eqref{def:si}
           \STATE $w_{i}^{(k)}=\max_{y \in \mathcal{S}_i}w(Ay+b)$
        \ENDFOR
         \STATE $M_{i}={\rm Median}(w_{i}^{(1)},w_{i}^{(2)},\dots,w_{i}^{(\ell)})$
       \ENDFOR 
       \STATE Return $M_{0}+(\frac{q}{r}-1)\sum_{i=0}^{n'-1}M_{i+1}\big(\frac{q}{r}\big)^{i} $
\end{algorithmic}
\end{algorithm}

%The theorem will be proved by a series of lemmas. 
To prove this theorem we borrow some ideas from coding theory. We define a linear $q$-ary code $C$ of dimension $n-m$ and length $n$ as the set of vectors $\{Ax: x\in \ff_q^{n-m}\}$ where $A$ is a full-rank matrix of size $n \times n-m$ and rank $n-m$. For a vector $a \in \ff_{q}^{n}$, we define the set $\{a+C\}$ as a coset of $C$. It is well known that  $\ff_{q}^{n}$ is partitioned by the $q^{m}$ distinct  cosets, each of size $q^{n-m}$. The main technique behind our algorithm is that for a random linear code $C$ of size $q^{n-m}$, we randomly sample $r^{m}$ distinct cosets of $C$. Subsequently, we find the maximum value $w(x)$ of an element among those $r^{m}$ cosets. 

Let $E \in \ff_q^{n \times n}$ be an $n\times n$ full rank matrix randomly and uniformly chosen from the set of all $n \times n$ rank-$n$ matrices over $\ff_q$. One can choose such a matrix via rejection sampling: independently and uniformly sample the entries of the matrix from $\ff_q$ and then reject the matrix and resample it if it is not full rank.  Let $A$ denote the random matrix formed by the first $n-m$ columns of $E$ as columns and let $R$ be the random matrix formed by the remaining $m$ columns of $E$ as columns. Also let $b$ be a vector sampled randomly and uniformly from $\ff_{q}^{n}$.
The MAX-oracle for \texttt{Unconstrained MB-WISH} is going to perform the following optimization   when $m \le n$:
\begin{align}
\max_{\sigma_1 \in \ff_q^{n-m}, \sigma_2 \in  \{\alpha_0,\alpha_1,\dots,\alpha_{r-1}\}^{m}} w(A\sigma_1+R\sigma_2 +b ).
\end{align}

Analogous to Theorem~\ref{thm:mb}, here we are creating union of $r^m$ distinct random bins. If we can prove that, for any element of $\ff_q^n$, the probability that it belongs to one of these bins is $(\frac{r}{q})^m$ and for any pair of different elements from $\ff_q^n$, whether they belong to one of these bins are independent (pairwise independence), the rest of the proof of Theorem~\ref{thm:mba} will  just follow that of Theorem~\ref{thm:mb}.

In particular, we just have to prove the lemma that is analogous to Lemma~\ref{lem:pair}.
Define a set
\begin{align*}
S_{A,R,b}\equiv \{Ax+b+Ry \mid x \in \ff_q^{n-m},y \in \{\alpha_0,\alpha_1,\dots,\alpha_{r-1}\}^{m}\}.
\end{align*}
For each configuration $\sigma \in \ff_{q}^{n}$, associate an indicator random variable $Z_{\sigma}$ denoting whether $\sigma \in S_{A,R,b}$ .

\begin{lemma}\label{lem:ach1}
For each configuration $\sigma \in \ff_{q}^{n}$, we must have $\Pr(Z_{\sigma}=1)=\Big(\frac{r}{q}\Big)^{m}$ and moreover for any two distinct configurations $\sigma_1,\sigma_2 \in \ff_{q}^{n}$, we must have $\Pr(Z_{\sigma_1}=1 \wedge Z_{\sigma_2}=1) \le (\Pr(Z_{\sigma}=1))^2$. 
\end{lemma}

\begin{proof}
Notice that $S_{A,R,b}$ is a union of distinct cosets and therefore,
\begin{align*}
S_{A,R,{\bf 0}} \equiv \bigcup_{y \in \{\alpha_0,\alpha_1,\dots,\alpha_{r-1}\}^{m}} S_{A,{\bf 0}}(y),
\end{align*}
where $S_{A,{\bf 0}}(y) \equiv  \{Ax+Ry \mid x \in \ff_q^{n-m}\}$ is defined as a particular coset with a fixed $y \in \{\alpha_0,\alpha_1,\dots,\alpha_{r-1}\}^{m}$. Hence $|S_{A,R,{\bf 0}}|=q^{n-m}r^{m}$ and since $S_{A,R,b}$ is simply a random affine shift of $S_{A,R,{\bf 0}}$, $|S_{A,R,b}|=q^{n-m}r^{m}$ as well. Now for a vector $\sigma \in \ff_q^{n}$, we must have
\begin{align*}
\Pr(Z_{\sigma}=1)=\sum_{y \in S_{A,R,b}} \Pr(\sigma=y)=\frac{|S_{A,R,b}|}{q^{n}}=\Big(\frac{r}{q}\Big)^{m}.
\end{align*} 
Next, for two configurations $\sigma_1,\sigma_2 \in \ff_{q}^{n}$, we have that
\begin{align*}
\Pr(Z_{\sigma_1}=1 \wedge Z_{\sigma_2}=1)&=\sum_{y_1,y_2} \Pr(\sigma_1 \in S_{A,b}(y_1) \wedge \sigma_2  \in S_{A,b}(y_2)) \\
&=\sum_{y_1,y_2} \Pr(\sigma_{1} \in S_{A,b}(y_1) \mid \sigma_{2} \in S_{A,b}(y_2))\Pr(\sigma_2 \in S_{A,b}(y_2)) \\
&=\sum_{y_1,y_2} \Pr(\sigma_{1}-\sigma_{2} \in S_{A,{\bf 0}}(y_{1}-y_{2}))\Pr(\sigma_2 \in S_{A,b}(y_2)). 
\end{align*}
Therefore we just need to evaluate the probability of the event $\Pr(\tau \in S_{A,{\bf 0}}(z))$ for $\tau=\sigma_1-\sigma_2 \neq 0$ and $z=y_1-y_2$. Now, if $z={\bf 0}$, $\Pr(\tau \in  S_{A,{\bf 0}}(\bf{0}))$ is equal to $\frac{q^{n-m}-1}{q^{n}-1}$ since $A$ is a randomly chosen full rank matrix, i.e.,  $|\{Ax:x \in \ff_q^{n-m} \}|\setminus \{0\}=q^{n-m}-1$. Now, since %the linearly independent columns of $R$ are sampled randomly. %,
 $\{Ax+Rz:x \in \ff_q^{n-m}\}, z \neq 0,$ is a uniformly random coset of $\{Ax:x \in \ff_q^{n-m} \}$, % over all the cosets (other than $\{Ax\}$) for $z \neq 0$. 
we have, 
\begin{align*}
&\Pr(\tau \in S_{A,{\bf 0}}(z) \mid z \neq 0, \tau  \in \{Ax:x \in \ff_q^{n-m} \})=0 \\
\textrm{ and }&\Pr(\tau \in S_{A,{\bf 0}}(z) \mid z \neq 0, \tau \notin \{Ax:x \in \ff_q^{n-m} \})=\frac{1}{q^{m}-1}.
\end{align*}
Hence,
\begin{align*}
\Pr( \tau \in S_{A,{\bf 0}}(z) \mid z \neq 0)&=%\Pr(\tau \notin Ax)\Pr(\tau \in Ax+Rz \mid z \neq 0,\tau \notin Ax)\\
\Big(1-\frac{q^{n-m}-1}{q^{n}-1}\Big)\frac{1}{q^{m}-1}=\frac{q^{n-m}}{q^{n}-1}.
\end{align*}
Therefore, we have that 
\begin{align*}
\Pr(Z_{\sigma_1}=1 \wedge Z_{\sigma_2}=1)&=\sum_{y_1,y_2: y_1=y_2}  \frac{q^{n-m}-1}{q^{n}-1}+\sum_{y_1,y_2: y_1 \neq y_2}  \frac{q^{n-m}}{q^{n}-1} \\
&= r^m\Big( \frac{q^{n-m}-1}{q^{n}-1}\Big)+\Big(r^{2m}-r^m\Big)\Big( \frac{q^{n-m}}{q^{n}-1}\Big) \\
&= \frac{r^m(r^mq^{n-m}-1)}{q^n-1} \\
& \le \Big( \frac{r}{q} \Big)^{2m} = \Pr(Z_{\sigma_1}=1)^{2}
\end{align*}
and hence we have the statement of the lemma.
\end{proof}
Although the two random variables $Z_{\sigma_1}$ and $Z_{\sigma_2}$ defined above are not independent, we show that they are {\em negatively correlated}. Note that, the pairwise independence was then subsequently used in computing a variance for the Chebyshev's inequality (see Lemma~\ref{lem:cru}). However, the negative correlation is sufficient to obtain an upper bound on the variance.

From Algorithm \ref{algo:mbawish} it is clear that Lemma \ref{lem:ach1} allows us to obtain the values of $M_i$ for $i \in \{1,2,\dots,n\}$.
Indeed, the MAX-oracle is not well defined when $m > n$.
 In order to obtain the values of $M_{i}$ for $i \in \{n+1,\dots,n'\}$, we propose the following technique.

Recall that the elements of $\ff_{q}^{n}$ can be represented as $n$ dimensional vectors where each element belongs to $\ff_{q}$. Moreover we defined an ordering over the elements of the finite field $\ff_{q}\equiv  \{\alpha_0, \alpha_1,\dots , \alpha_{q-1}\}$ so that $\alpha_{i} \prec \alpha_{j}$ for $i <j$. 
Consider the lexicographic ordering of the elements (vectors) of $\ff_{q}^{n}$.
%We can extend this ordering over the elements of $\ff_{q}^{n}$ so that for two elements $x,y \in \ff_{q}^{n}$, $x <y$ if $x_j <y_j$ where $j$ is the smallest index in the vector representation of $x$ and $y$ for which $x_j \neq y_j$. \par
Let $s_m$ be the $\lceil \frac{r^{m}}{q^{m-n}} \rceil$ th element in this ordering of $\ff_{q}^{n}$. Define the set 
\begin{align}\label{def:si}
\mathcal{S}_m=\{x \in \ff_{q}^{n} \mid x \prec s_m \} 
\end{align}
for all $m>n$. 
Now, let $A \in \ff_q^{n \times n}$ be an $n\times n$ full rank matrix randomly and uniformly chosen from the set of all $n \times n$ rank-$n$ matrices over $\ff_q$, which can be generated by rejection sampling as before.
%Generate a random uniform full rank matrix $A \in \ff_{q}^{n \times n}$ and 
Let  $b \in \ff_{q}^{n}$ be a uniform random vector. Subsequently, the MAX-Oracle for \texttt{Unconstrained MB-WISH} solves the following optimization problem for $m>n$:
\begin{align*}
\max_{y\in \mathcal{S}_m} w(Ay+b).
\end{align*}
In order to analyze the statistical properties of this oracle, define the random set
\begin{align*}
T_{A,b,m} \equiv \{Ay+b \mid y \in \mathcal{S}_{m}\}.
\end{align*}	
Again, for each configuration $\sigma \in \ff_{q}^{n}$, associate an indicator random variable $Z_{\sigma}$ denoting $\sigma \in T_{A,b,m}$.

\begin{lemma}
For each configuration $\sigma \in \ff_{q}^{n}$, we must have $\Big(\frac{r}{q}\Big)^{m} -\frac1{q^n}\le \Pr(Z_{\sigma}=1)\le  \Big(\frac{r}{q}\Big)^{m}$ and moreover for any two configurations $\sigma_1,\sigma_2 \in \ff_{q}^{n}$,  $\Pr(Z_{\sigma_1}=1 \wedge Z_{\sigma_2}=1) \le (\Pr(Z_{\sigma}=1))^2.$ 
\end{lemma}
\begin{proof}
We have, 
\begin{align*}
\Pr(Z_{\sigma}=1)=\sum_{y \in T_{A,b,m}} \Pr(\sigma=y)=\frac{|T_{A,b,m}|}{q^{n}}=\frac{1}{q^n}\Big\lfloor \frac{r^{m}}{q^{m-n}} \Big\rfloor, % \le  \Big(\frac{r}{q}\Big)^{m}.
\end{align*} 
which proves the first claim.
Next, for two distinct configurations $\sigma_1,\sigma_2 \in \ff_{q}^{n}$, we have that
\begin{align*}
\Pr(Z_{\sigma_1}=1 \wedge Z_{\sigma_2}=1)&=\sum_{y_1,y_2 \in T_{A,b,m}} \Pr(\sigma_1=y_1 \wedge \sigma_2=y_2) \\
&=\sum_{y_1,y_2 \in \mathcal{S}_{m}} \Pr(\sigma_1=Ay_1+b \wedge \sigma_2=Ay_2+b) \\
&=\sum_{y_2 \in \mathcal{S}_{m}} \Pr(\sigma_2=Ay_2+b)\sum_{y_1 \in \mathcal{S}_m} \Pr(\sigma_1=Ay_1+b \mid \sigma_2=Ay_2+b) \\
&=\sum_{y_2 \in \mathcal{S}_{m}} \Pr(\sigma_2=Ay_2+b)\sum_{y_1 \in \mathcal{S}_m} \Pr(\sigma_1-\sigma_2=A(y_1-y_2)) \\
\end{align*}
Since $\sigma_1 \neq \sigma_2$, we must have that $\Pr(\sigma_1-\sigma_2=A(y_1-y_2) \mid y_1=y_2)=0$. For $y_1 \neq y_2$, every configuration $\sigma \in \ff_{q}^{n}, \sigma \neq 0$ is equally probable to be $A(y_1-y_2)$ since $A$ is uniformly and randomly sampled full rank matrix. Hence,
\begin{align*}
\Pr(Z_{\sigma_1}=1 \wedge Z_{\sigma_2}=1)&=\frac{1}{q^{n}(q^n-1)}\Big\lfloor \frac{r^{m}}{q^{m-n}} \Big\rfloor \Big(\Big\lfloor \frac{r^{m}}{q^{m-n}} \Big\rfloor-1\Big) \le (\Pr(Z_{\sigma}=1))^2.
\end{align*}
%Even if we do not make the above approximation, we can observe that $Z_{\sigma_1}$ and $Z_{\sigma_2}$ are negatively correlated which is only stronger.
\end{proof}

The remainder of the proof of Theorem~\ref{thm:mba} follows that of Theorem \ref{thm:mb} in a straightforward manner.

%As we will see in the next section, it is possible to reduce the memory requirement for our algorithm by using randomness-optimal hash families.

%{\bf Show some examples (say $q=3$); end with a remark on either k-wise independence, or $(1+\epsilon)$ approximation via product spaces, and/or the utility of the $q$ ary thing (say Potts model) We have shown approximation improvements schemes using multiple bins with fields of higher size. It allows us to improve the approximation and also allows us to extend Ermon's work to configurations which are defined over $F_{q^{n}}$ without having to rename the configurations. Previously Ermon's work was not able to provide a better approximation factor than 16 without resorting to creating problem of bigger size (Please see discussion in the section Approximation Improvement for more details). Also optimizing the number of random bits is possible here.}

\section{Experimental results} \label{sec:exp}
All the experiments were performed in a shared parallel computing environment that is equipped with 50 compute nodes with 28 cores Xeon E5-2680 v4 2.40GHz processors with 128GB RAM. %Further experiments on estimating the TV distance is reported  in the supplementary material.

\paragraph{Experiments on simulated Potts model (regular degree graph).}
 %We implemented our algorithm to estimate the partition function of  Potts Model  (cf.~eq.~\eqref{eq:part}). For simulations, we generated random networks using a python library \texttt{networkx} used a python module \texttt{constraint} (for Constraint Satisfaction Problems (CSPs)) to handle the constrained optimization for MAX-oracle. 
%The detailed results of this experiment is provided in Appendix~\ref{app:tv} in the supplementary material. 
%The worst approximation factor that we saw by running \texttt{MB-WISH} is $5.44$ (see Table~\ref{tab:part}). % in the supplementary material).
We implemented our algorithm to estimate the partition function of  Potts Model. Recall that the partition function of the  Potts model on a graph $G=(V,E)$   is given in Eq.~\eqref{eq:part}. First of all, we computed partition functions for small graphs where a brute-force algorithm can also be used.
%\begin{align*}
%Z(G) & =\sum_{\sigma }\exp (-\zeta h(\sigma))\\
%\intertext{where}
%-h(\sigma)  &=H\sum_{v_i \in V}\delta_{Kr}(\sigma_i,1)+ J\sum_{(v_i,v_j) \in E}\delta( \sigma_i,\sigma_j)
%\end{align*}  
%Here the $v_i$ are the vertices  of $G$ and $(v_i,v_j) $ its edges; the $\sigma \in \{0,1,\ldots ,q-1 \} $ are the spin variable at $v_i$. $J$ is the spin coupling and $H$ is the external force. $\beta$ is parameter that dependent on temperature of the system and the sum is  over all states $\sigma$ of $G$.\\
For our simulation, we have randomly generated the graph $G$ with number of nodes $n \equiv |V|$ varying in $4,5,6,7,8,9,$ and corresponding regular degree $d=2,2,4,4,4,4,$ using a python library \texttt{networkx}. We took the number of states of the Potts model $q=5$, the external force $H$ and the spin-coupling $J$   to be 0.1 and then varied the values of $\zeta$.
The partition functions for different cases are calculated using both brute force and our algorithm (\texttt{MB-WISH}). We have used a python module \texttt{constraint}  to handle the constrained optimization for MAX-oracle. The obtained approximation factors for different  $\zeta$ are listed in Table~\ref{tab:part}. The worst approximation factor observed in all these trials is $5.442$.
This experiment shows that, for small graphs the partition functions computed by \texttt{MB-WISH} are good approximations to the actual values. 

\remove{
For $n=10, 11, d=6$ the approximation factor for \texttt{MB-WISH} is exactly 1 (up to the precision of the number system used). However the time taken by \texttt{MB-WISH} is $30$ minutes for $n=10$ and two hours for $n=11$ which is much higher than the time required in the results for Table~\ref{tab:part}.
For $n=12, d=8$, \texttt{MB-WISH} gives an approximation factor of $2.5$ after running for eight hours in the above  parallel computing environment.
}
\begin{table*}[t]
\begin{center}
	\begin{tabular}{| c| c |c |c | c |c|c|  }
		\hline
$\zeta $ & $n=4, d=2$ & $n=5,d=2$ & $n=6,d=4$ & $n=7,d=4$ & $n=8,d=4$ & $n=9,d=4$ \\
\hline
 0 & 0.976 & 1.220 & 0.610 & 1.907 &  0.953 & 1.192\\
 \hline
 5& 0.580 & 0.708 & 1.639  &  0.755 &  0.630 & 0.599\\
\hline
 10 & 0.7470 & 1.191 & 3.271 & 0.989 & 1.875 &  1.25\\
\hline
 15 &   1.430 & 1.036 & 1.013 & 1.224 & 1.399 & 1.692\\
 \hline
 20 &  1.032 & 1.590 & 1.141 & 1.173 & 1.365 & 1.491\\
 \hline
 25 & 0.839 & 1.118 & 1.339 & 1.035 & 1.429 & 1.326\\
 \hline
 30 & 0.510 & 4.0562 & 2.226 & 1.060 & 0.690 & 2.122 \\
 \hline
 35 & 1.073 & {\bf 5.442} & 0.489 & 2.871 & 1.639 & 1.263\\
 \hline
 40 &  1.210 &  2.434 &  0.980 &  0.582 &  0.666 & 0.969 \\
 \hline
 45 & 1.127 &  4.640 &  2.348 &  1.336&  0.3673 & 1.341\\
 \hline
 50 &  1.152 &  1.025 &  2.511 &  3.4307 &  1.1522 &2.636 \\       
\hline
	\end{tabular}
\caption{The ratio of the partition function calculated by \texttt{MB-WISH} ($r=2$) and the actual value calculated by brute force: $\frac{\hat{Z}}{Z}$. \label{tab:part}}
\end{center} 
\end{table*}

\begin{table}[h!]
\begin{center}
	\begin{tabular}{|c|c|c|c|c|c|c|c|c|c|}
		\hline
		\multirow{2}{*}{$n$}	& \multicolumn{3}{c|}{$\zeta=1$} &  \multicolumn{3}{c|}{$\zeta=2$} &    \multicolumn{3}{c|}{$\zeta=5$} 	\\
		\cline{2-10}
		 & MB-WISH  & BP & MCMC & MB-WISH  & BP & MCMC & MB-WISH & BP & MCMC\\
		\hline
		10 & 15.16 & 15.51 & 12.60 & 14.35 & 14.98 & 12.06 & 13.07 & 13.56 & 10.65 \\
		\hline
         15 &	23.10 & 23.27 &	 20.51 & 22.35 & 22.47 & 19.70 & 19.95 & 20.35 & 17.59 \\
         \hline		
		20 & 31.04 & 31.03 & 28.69 & 29.98 & 29.96 & 27.62 & 26.93 & 27.13  & 24.80 \\
		\hline
		25 &  38.28 & 38.79 & 36.63 & 37.41 & 37.45 & 35.29 & 33.41 &  33.92 & 31.76 \\
		\hline
		30 &  46.23 & 46.55 & 44.49 & 44.51 & 44.94 & 42.89 & 40.82 & 40.705 & 38.65 \\
		\hline
		40 &  61.88 &  62.06 & 59.75 & 59.55 & 59.92 & 57.61  & 54.96 & 54.27 & 51.96 \\
		\hline
		50 &  77.31 & 77.58 & 75.28 &  74.69 & 74.90 &  72.59 &  68.62 & 67.84 & 65.54 \\
		\hline
	\end{tabular}
\end{center}
\caption{\small Log-partition function computed by \texttt{unconstrained MB-WISH}, Belief Propagation (BP) and Markov Chain Monte Carlo (MCMC) respectively for the cases of $\zeta=1,2$ and $5$. } \label{tab:part1}
\end{table}

For graphs with larger number of vertices, it is not possible to compute the partition function of Potts Model by brute force. Therefore, we compare the partition function computed by \texttt{Unconstrained MB-WISH} ($\hat{Z}$) with 
two standard techniques: Belief propagation (BP)  \cite{koller2009probabilistic} and Markov-Chain-Monte-Carlo (MCMC) \cite{jerrum1996markov}. It is known that BP provides exact result when the underlying graph is cycle-free \cite{koller2009probabilistic}. To implement this we use the PGMPY library in \texttt{python} \cite{PGMPY}. For MCMC,  we employ the popular Metropolis-Hastings (MH) algorithm \cite{kroese2011handbook} to sample random points from $\Omega$, where we evaluate the function $w:\Omega\to \reals$ and take a scaled-sum to estimate the discrete integration problem. We have calculated the average of the partition function over 10 different trials of the MH algorithm, and  each trial was given the same time as that of \texttt{Unconstrained MB-WISH}.

%the one ($Z$) computed by a belief propagation algorithm in the in the PGMPY library in \texttt{python} \cite{PGMPY}.
%The following experiment has been performed to estimate the partition function of Potts Model on  a Graph $G=(V,E)$ . 
Again, for our simulation, we have randomly generated the graph $G$ with number of nodes $n \equiv |V|$ varying in $10,\dots,50,$ and with regular degree $d=4$ using a python library \texttt{networkx}.
We took the number of states of the Potts model $q = 5$, the external force $H$ and the spin-coupling   to be   $ 0.1$ and then varied the values of  $\zeta $. In our experiments each optimization  instances are run with a timeout of $10,15 ,20,20,25 $ minutes for $n=20,25,30,40,50$ respectively (we let the $n=10$ case run without a time constraint). The results are summarized in Table \ref{tab:part1}. 
 It can be observed that the partition functions computed with MCMC deviate somewhat from that computed with belief propagation, whereas \texttt{MB-WISH} gives  values closer to the belief propagation results.

Since for cycle-free graphs, BP can provide exact result, it gives an opportunity to compare \texttt{MB-WISH} with the single-bin version (i.e., Ermon et al.'s original algorithm)  for moderate values of $n$ and $q$. We perform the next experiment on a path-graph, which is an undirected graph where there are exactly two nodes of degree $1$ and every other node has degree $2$.
We perform the experiment with the number of nodes $n \equiv |V|$ varying in $20,\dots,50$ on a path-graph such that the number of states $q=31$ and the external parameters $J=0.1$, $H=0.5$ and $\zeta=-5$. For two different values of $r$, respectively $1$ (single-bin) and $15$ (multi-bin) we compute the estimates of the partition function. We have plotted the ratio of the estimates with the corresponding ones computed by BP (which is exact), in Figure~\ref{fig:chain}.
It is clear from the figure and the table that the \texttt{Unconstrained MB-WISH} performs much better than its single-bin counterpart. 
The timeout for each call to the oracle is chosen to be $n/10$ where $n$ is the number of nodes in the graph.
\begin{figure}[htbp]
  \centering
   \includegraphics[width=\textwidth]{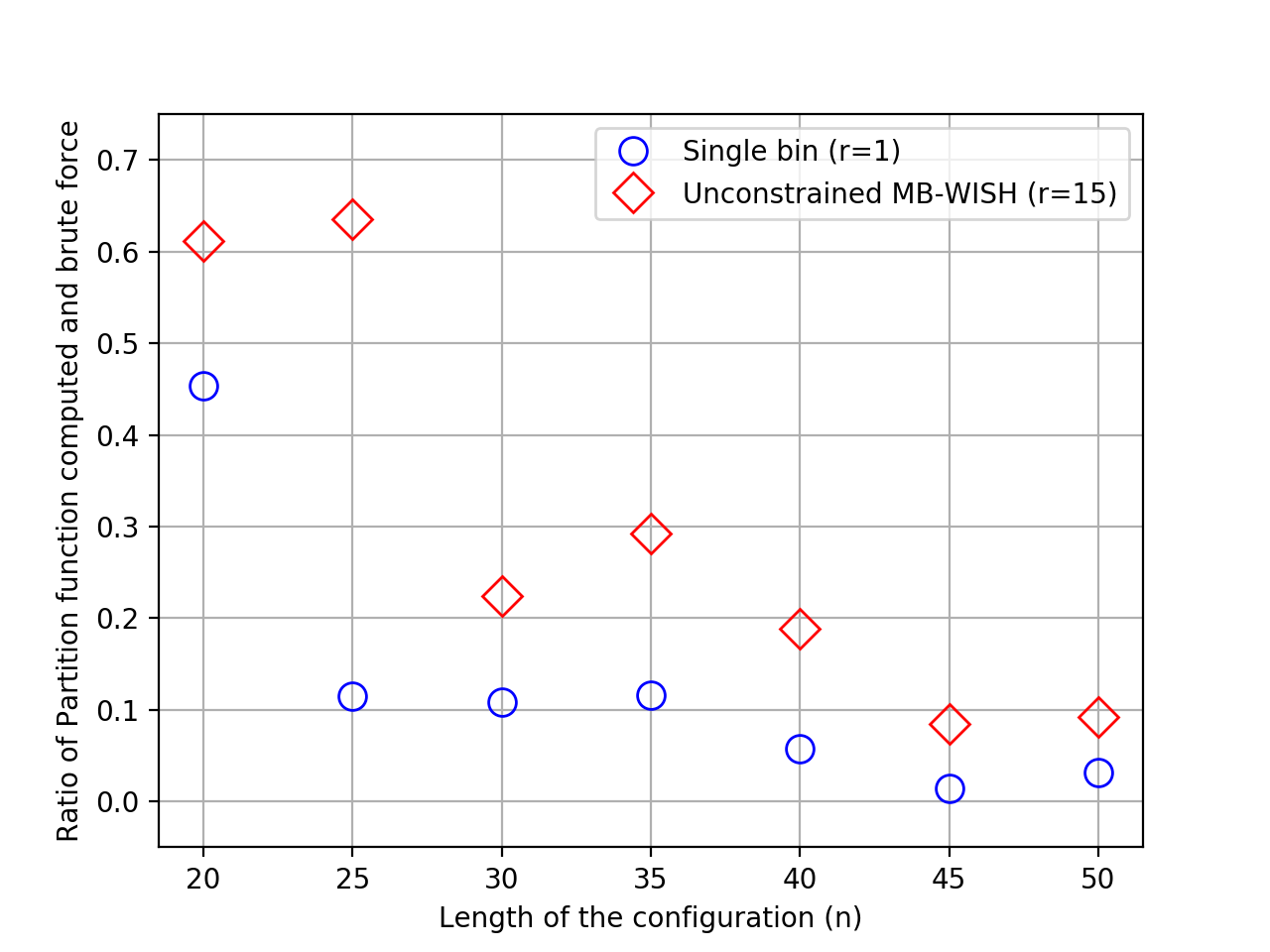}
    \caption{Comparison of approximation ratios obtained by using \texttt{Unconstrained MB-WISH} (red) and single-bin (Ermon et al.'s method) (blue). A ratio closer to 1 is better.}
    \label{fig:chain}
 \end{figure}

\begin{figure}[htbp]
  \centering
  \begin{subfigure}{0.45\textwidth}
    \includegraphics[width=\textwidth]{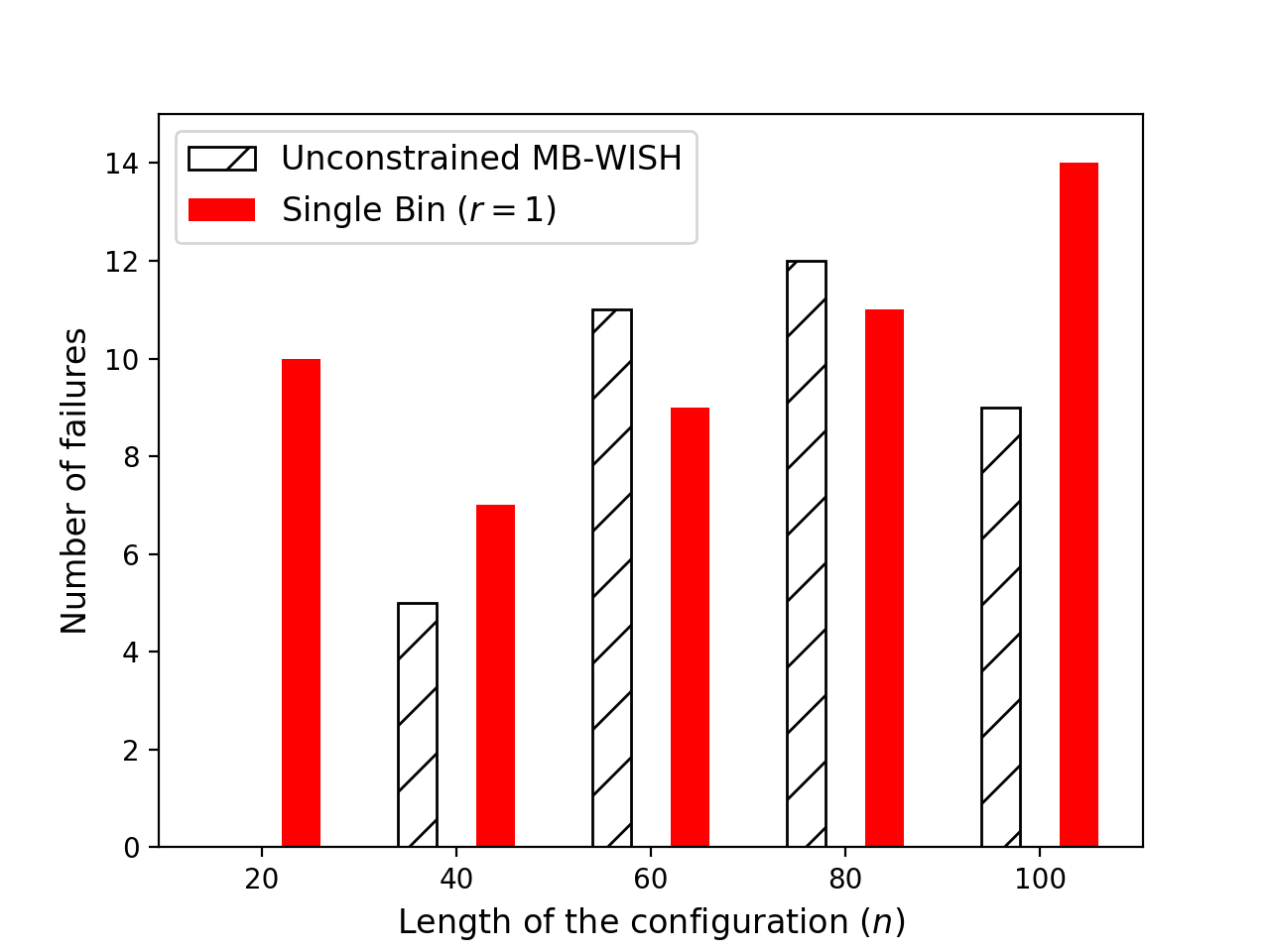}
    \subcaption{Number of times, among 100 trials,  the computed total variation distance is above $4 \times$ the theoretical upper bound for $\epsilon=10^{-2}$}
    \label{fig:ub}
  \end{subfigure}~~
  \quad
  \begin{subfigure}{0.45\textwidth}
    \includegraphics[width=\textwidth]{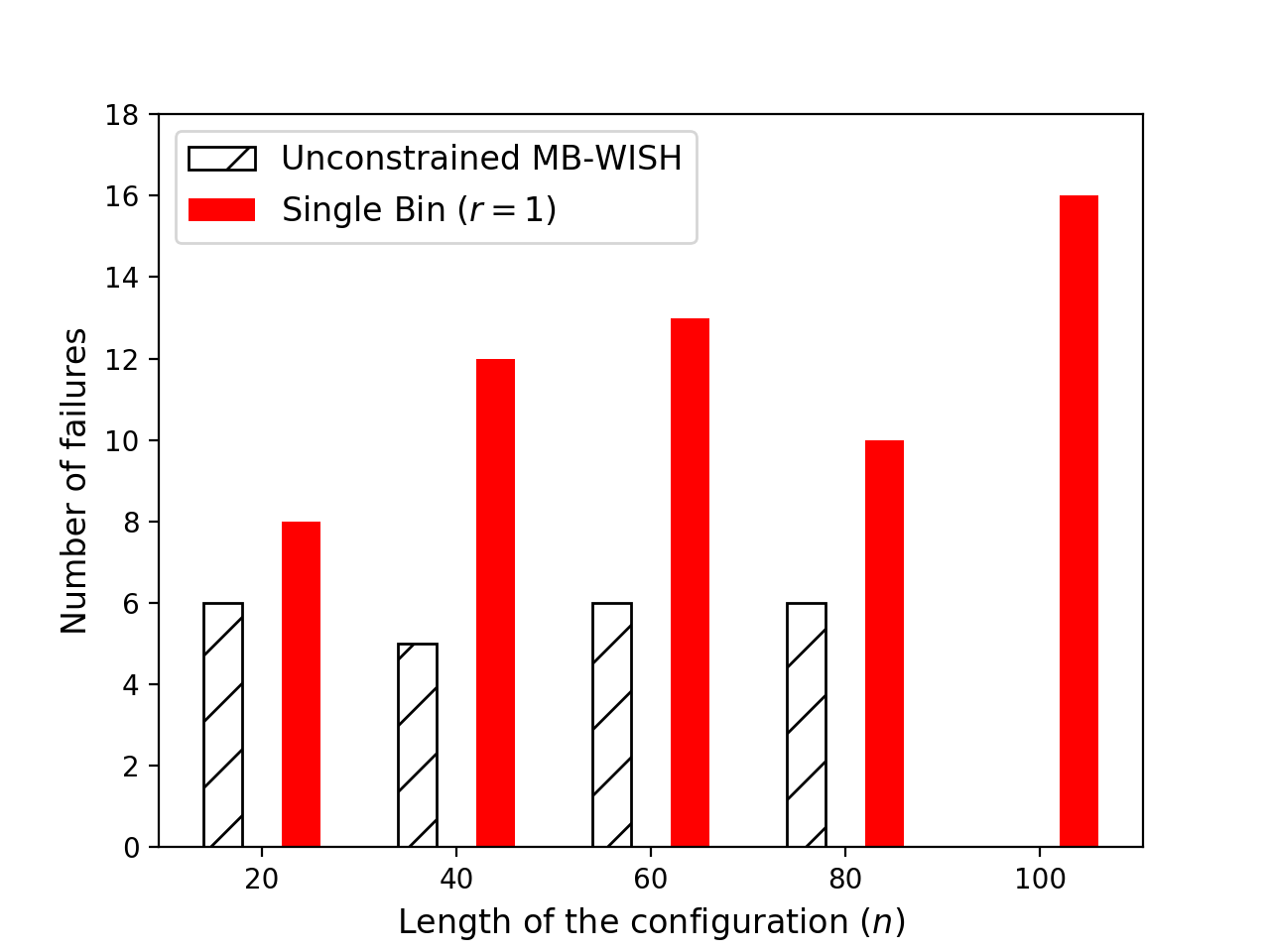}
    \subcaption{Number of times, among 100 trials,  the computed total variation distance is above $4 \times$ the theoretical upper bound for $\epsilon=10^{-4}$}
    \label{fig:ub2}
  \end{subfigure}~~
  \quad \\
  \begin{subfigure}{0.45\textwidth}
    \includegraphics[width=\textwidth]{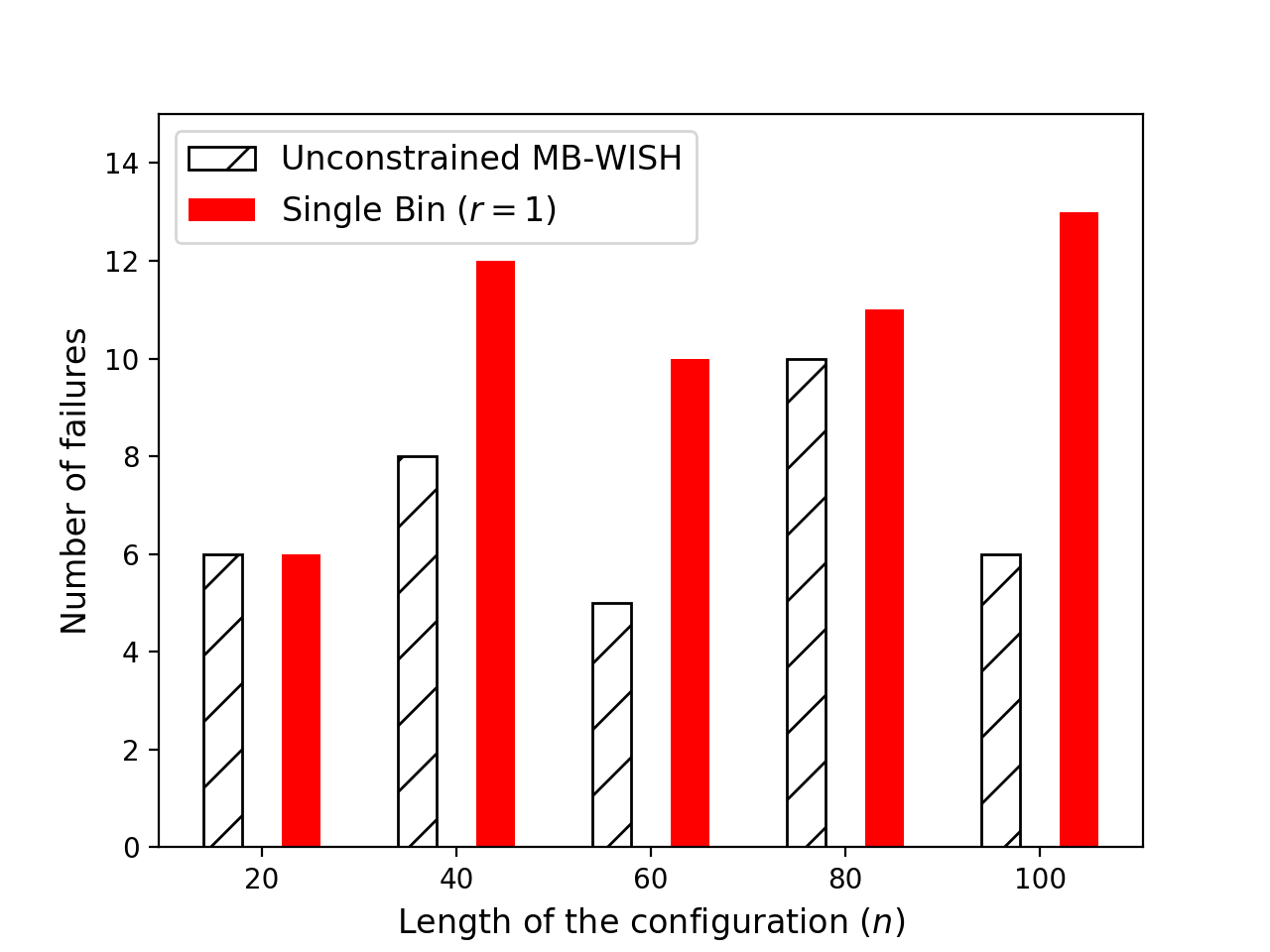}
    \subcaption{Number of times, among 100 trials,  the computed total variation distance is above $4 \times$ the theoretical upper bound for $\epsilon=10^{-6}$}
    \label{fig:ub4}
  \end{subfigure}~~
  \quad
   \begin{subfigure}{0.45\textwidth}
    \includegraphics[width=\textwidth]{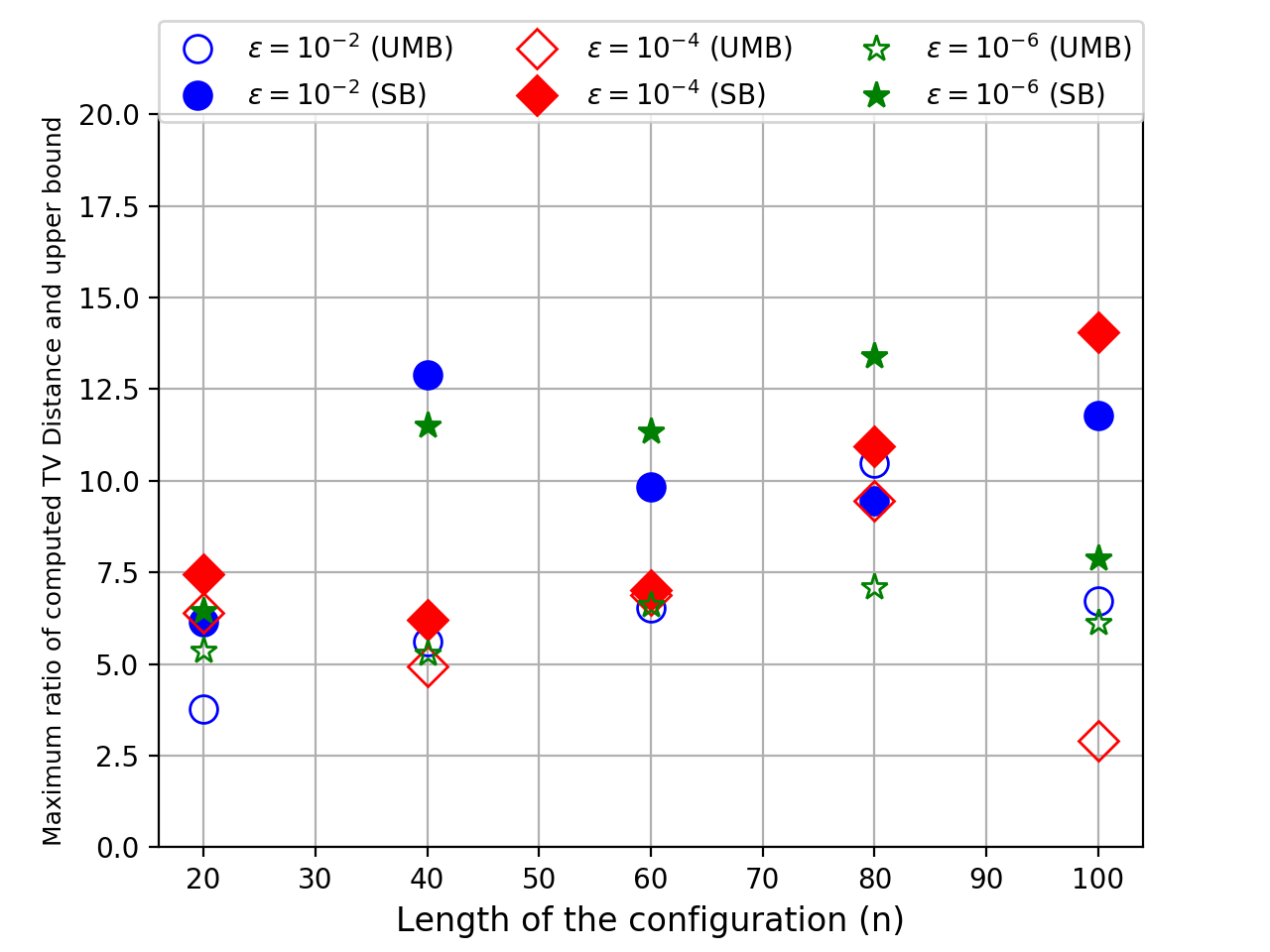}
    \subcaption{The maximum ratio of the computed total variation distance and the upper bound}
    \label{fig:max}
  \end{subfigure}~~
  \caption{Behavior of computed total variation distance by Algorithm \ref{algo:mbawish} using $r=2$ (Unconstrained MB-WISH or  UMB) and $r=1$ (Single Bin or SB) with respect to the theoretical upper bound.\label{fig:tv}}
\end{figure}

\paragraph{Experimental results on computing Total Variation distance.}
We show one more instance of discrete integration where \texttt{MB-WISH} is useful. The purpose of this experiment is to show the effectiveness of \texttt{MB-WISH} by choosing counting problems where good theoretical results are available.

For this, we use \texttt{Unconstrained MB-WISH} to compute {\em total variation distances} between two high dimensional (up to dimension $100$) probability distributions, generated via an iid model. Although, it is computationally hard to compute the total variation distance between two distributions, for the special case of product distributions, we can derive theoretical expressions that are known to bound the total variation distance from above and below. 

%The detailed results are provided next. %in Appendix~\ref{app:tv} in the supplementary material.

 The total variation (TV) distance between any two discrete distributions $P$ and $Q$ with common sample space $\Pc$ is defined to be
$$
\|P-Q\|_{TV} = \sup_{A\subseteq \Pc}|P(A) -Q(A)| = \frac12 \sum_{\sigma \in \Pc}|P(\sigma) -Q(\sigma)|.
$$
%It is difficult to estimate the TV distance between two distribution over product spaces. Even if the random variables are independent, there is no easy way to analytically compute the total variation distance. 
Simply consider finding TV distance between joint distributions of $n$ random variables that can take value in $\{0, 1, \dots, q-1\}$. In that case, we seek to find,
$$
\frac12 \sum_{\sigma \in \{0,1, \dots, q-1\}^n}|P^{n}(\sigma) -Q^{n}(\sigma)|,
$$
which is in the exact form of Eq.~\eqref{eq:main}.  Therefore we can use \texttt{MB-WISH} algorithm to estimate the total variation distance. The following are well-known upper and lower bounds on TV distance based on Hellinger distance, $h(P,Q)^2 \equiv \sum_{\sigma \in \Pc}(\sqrt{P(\sigma)}-\sqrt{Q(\sigma)})^2$:\footnote{See \url{https://stanford.edu/class/stats311/Lectures/full_notes.pdf}.}
$$
\frac12 h(P,Q)^2\leq \|P-Q\|_{TV}\leq h(P,Q)\sqrt{1- h(P,Q)^2/4}.
$$
Furthermore,  %It is known that,
\begin{align*}
 h(P^n,Q^n)^2 &=2-2\prod_{i=1}^{n}(1-\frac{1}{2}h(P_i,Q_i)^2).
\end{align*}
For `near-uniform' distributions, it is known that the upper bound is a good approximation~\cite{sason2016f}. 
%where $h(P,Q)^2 = 1- \sum_{\sigma \in \Pc}\sqrt{P(\sigma)Q(\sigma)}$.

%The total variation distance for any distribution $P$ and $Q$ can be bounded as follows:
%\begin{enumerate}
%\item  For the Hellinger distance :
%	\begin{align*}
%	\frac{1}{2}h(P,Q)^2\leq TV(P,Q)\leq h(P,Q)\sqrt{1- h(P,Q)^2/4}\leq h(P,Q)
%	\end{align*}
%	\item Pinsker Inequality
%	\begin{align*}
%	TV(P,Q)\leq \sqrt{\frac{1}{2}D_{KL}(P||Q)}
%	\end{align*}
%\end{enumerate}

For the experiments, we choose two distributions defined over $q$ points in the following manner: We choose a vector $\mathbf{v} \in [0,1]^{q}$ randomly and normalize the vector (so that the sum of the elements is $1$) in order to have the first distribution $P \equiv [p_1,p_2,\dots,p_q]$. The second distribution $Q$ is  then chosen to be 
\begin{align*}
Q \equiv [p_1,p_2+\epsilon,p_3-\epsilon,\dots,p_q-\epsilon].
\end{align*} 
where $\epsilon$ is a small number chosen in order to make the two distributions very close to each other.  Here the distribution $P^n$ and $Q^n$ are supported on $\{0,1,2,\dots,q-1\}^n$ where $n$ can be any natural number. Now, we choose $q=5$ and for three different values of $\epsilon=10^{-2}, 10^{-4}$ and $10^{-6}$ and five different values of $n=20,40,60,80$ and $100$, we repeat the experiment described above $100$ times for each setting and use \texttt{Unconstrained MB-WISH} to compute the total variation distance.%The bound can be expressed as:
%\begin{align*}
%TV(P^n,Q^n)& \leq \sqrt{\frac{n}{2}D_{KL}(P||Q)}\\
%TV(P^n,Q^n)&\leq  h(P^n,Q^n)\leq \sqrt{2-2(1-\frac{1}{2}h(P_i,Q_i)^2)^n}
%\end{align*}

%To compare our method with the upper and lower bound with Hellinger distance we choose the random variables to be independent, although that our method is capable of approximation irrespective of any such conditions.

 We perform our experiments in a time constrained manner (10 minute for each calls to MAX-oracle). We have shown in Figure \ref{fig:tv} histograms of the number of times the computed total variation distance is above four times the upper bound for $\epsilon=10^{-2}, 10^{-4}$ and $10^{-6}$ respectively in Figures \ref{fig:ub}, \ref{fig:ub2} and \ref{fig:ub4} respectively. 
%We consider the computed total variation distance to be a failure if the value is above $4$ times the theoretical upper bound. 
We chose a factor of four because the theoretical approximation factor guaranteed by \texttt{Unconstrained MB-WISH} is $\sim 4$.
We observed that the total variation distance is always above the upper bound with Hellinger distance but on the other hand, in very few trials the computed value is above four times the upper bound.  Finally, we have also shown the maximum ratio of the computed total variation distance and the upper bound for each value of $\epsilon$ and $n$ in Figure \ref{fig:max}. 

We have compared the results obtained by \texttt{Unconstrained MB-WISH} with the corresponding results obtained by its single bin counterpart (Ermon et al.'s method, choosing $r=1$) in Figure \ref{fig:tv}. Even though $q=5$ is not large, the improvement in performance by using \texttt{Unconstrained MB-WISH} is clear. In Figures \ref{fig:ub}, \ref{fig:ub2}, \ref{fig:ub4}, it can be observed that in the case of single bin, the number of failures (solid red) is almost always larger than the corresponding setting with multiple bins. Moreover, in Figure \ref{fig:max}, the maximum ratio in the setting of single bin (solid) is always much higher than the the setting of multiple bins (hollow).

\paragraph{Real-world constraint satisfaction problem (CSPs).} 
 Many  instances of real-world graphical models are available in 
%\texttt{http://www.csplib.org/} or  
\url{http://www.cs.huji.ac.il/project/PASCAL/showExample.php}.
Notably, some of them (e.g., image alignment, protein folding) are defined on non-Boolean domains, which justify the use of \texttt{MB-WISH}. We have computed the partition functions for several of them. %datasets from this repository. 

The  dataset \texttt{Network.uai} 
is a Markov network with $120$ nodes each having a binary value. A configuration here is a binary sequence of length $120$. To calculate the partition function, we need to find the sum of weights for $2^{120}$ different configurations. In order to use \texttt{Unconstrained MB-WISH}, we view each configuration as a $16$-ary string of length $30$.  Our results for the log-partition came out to be $156.00$ with one hour time out for each call to the MAX-oracle. The benchmark for the log-partition function is provided to be $163.204$.

%Log-partition functions of some other datasets from this repository is provided in Table \ref{tab:real}.
%
%\begin{table}
%\begin{center}
%	\begin{tabular}{|c|c|c|c|c|c|c|}
%		\hline
%		Dataset Name & Genome & 
%		$n$	 & $5$ & $6$ & $7$ & $8$ & $9$ & $10$\\
%			\hline
%			$\hat{Z}/Z$ &	$2.137$	&$1.103$& $1.245$&  $2.407$ & $0.730$ & $1.527$\\
%\hline
%	\end{tabular}
%\end{center}
%\caption{Approximation factor of Permanent  of an  $n \times n$ matrix via \texttt{MB-WISH}. \label{tab:perm}}
%\end{table}

%\begin{table}[h!]
%\centering
%	\begin{tabular}{|c|c|c|c|}
%		\hline
%		$n$	&$\zeta=1$ & $\zeta=2$ &   $\zeta=5$ 	\\
%			\hline
%		$10$ & $0.054273$   & $0.054390$     &$ 0.054318$      \\ 
%		\hline
%		$15$ & $0.063036$   & $0.0631594$    &$0.0633923$     \\ 
%		\hline
%	$20$ &  $0.09653491$  &  $0.096632$   & $0.09662014$      \\ 
%		\hline
%		$25$ & $0.115449$   & $0.1159425$     &$0.116080$      \\ 
%		\hline
%		$30$ & $0.1285637$   & $0.12862$    &$0.128368$     \\ 
%		\hline
%		$40$ & $0.09931$   & $0.0995325$    &$0.099551$     \\ 
%		\hline
%		$50$ & $0.0998276$   & $0.099575$    &$0.09992$     \\ 
%		\hline
%	\end{tabular}
%\caption{\small The ratio  of partition function computed with MCMC to that computed with belief propagation. \label{tab:mcmc}}
%\end{table}

The  \texttt{Object detection} dataset comprised of $60$ nodes each having a $11$-ary value and by \texttt{Unconstrained MB-WISH} we found the log-partition function to be $-38.9334$.  The \texttt{CSP} dataset is a Markov network with $30$ node having a ternary value: we found the log partition function to be $-39.9933$.  For these datasets there were no baselines available for comparison. The purpose of these experiments were to establish the scalability of \texttt{MB-WISH}.

\section{Conclusion}
Large scale counting problems (or discrete integrations of  nonnegative weight functions) are often computationally intractable, but come up frequently in variety of inference tasks, most prominently as evaluations of partition functions. In this paper we extend a recent technique of hashing and optimization due to Ermon et al. for discrete integration over hypercube $\{0,1\}^n$ to that over hypergrids $\{0,1, \dots, q-1\}^n$. The trivial generalization results in an approximation factor that rapidly becomes worse as $q$ increases. We remedy the situation by providing constant factor approximation algorithms for all $q.$
 
The main drawback of this approach of discrete integration is the delegation of a hard combinatorial optimization to an oracle. In this line of work, an open problem is to come up with hash functions that maintain the essential properties (such as pairwise independence), but make the oracle optimization amenable. While in general this is not possible, for certain classes of weight functions this may be a plausible task and requires further exploration.

%\appendix

%\begin{center}
%
%{\Large High Dimensional Discrete Integration by Hashing and Optimization\\ Appendix}
%
%\end{center}
%\section{Omitted proofs from Section~\ref{sec:wish}}\label{sec:proof}

%\vspace{0.1in} 

%\vspace{0.2in}

\bibliographystyle{plain}
\bibliography{references}

\appendix
\section{Derandomization: structured hashes}\label{sec:der}
\remove{Finally, {\em We show that it is possible to come up with hash functions of the same form (i.e., $Ax+b$) and same pairwise independence properties using only $m+n$ random bits.} %Moreover we can use a (random) Toeplitz matrix in place of the matrix $A$. 
If one uses a random sparse Toeplitz matrix, the construction of hash family takes only $O(n)$ random bits and the hash operation can be faster because of the structure of the matrix.}

For the analysis of  \cite{ermon2013taming,ermon2014low} to go through, we needed a family of hash functions that are pairwise independent\footnote{It is sufficient to have the hash family satisfy some weaker constraints, such as being pairwise negatively correlated.}. A hash family  $\Hc = \{h:\Omega \to \tilde{\Omega}\}$ is called uniform and pairwise independent if the following two criteria are met for a randomly and uniformly chosen $h$ from $\Hc$:
%\begin{itemize}
1) for every $x \in \Omega,$ $h(x)$ is uniformly distributed in $\tilde{\Omega}$ and 2) for any two $x,y \in \Omega$ and $u,v \in \tilde{\Omega}$, $\Pr(h(x) =u, h(y) =v) = \Pr(h(x) =u)\Pr(h(y) =v).$ 
%\end{itemize} 
By identifying $\Omega$ with $\ff_2^n$ (and $\tilde{\Omega}$ with $\ff_2^m$) and by using a family of hashes $\{x\mapsto h_{A,b}(x) = Ax + b: A \in \ff_2^{m \times n}, b \in \ff_2^m\}$ defined in \eqref{eq:hash}, \cite{ermon2013taming} show the family to be pairwise independent and thereby achieve their objective.

 The size of the hash family $\Hc$ determines how many random bits are required for the randomized algorithm to work. By defining the hash family by a random binary matrix, Ermon et al. reduce the number of random bits from potentially $m2^n$ to $mn +m = m(n+1)$ bits (see, p.~3 of \cite{ermon2013taming}).
Here, we show that it is possible to construct pairwise independent hash family $\{\ff_2^n\to \ff_2^m\}$ using only $O(n)$ random bits such that any hash function from the family still has the structure $h(x) = Ax+b$. 
While memory optimal pairwise independent hash functions are quite standard, we feel for completeness it would be good to show that they can be represented as the above matrix-vector product form.
All of the statements of this section can be easily extended to $q$-ary alphabets. % The proofs of the  propositions below can all be found in Appendix~\ref{app:der} in the supplementary material.

\medskip

\noindent{\bf Construction 1:} %For the first construction, we need to consider  choose 
Let $f(x) \in \ff_2[x]$ be an irreducible polynomial of degree $n$. We construct the finite field $\ff_{2^n}$ with the $\zeta$, root of $f(x)$ as a generator of $\ff_{2^n}^\ast$. Now, any $x \in \ff_2^n$ can be written as a power of $\zeta$ via a natural map $\phi: \ff_2^n \to \ff_{2^n}$. Indeed, for any element $\zeta^{k} \in \ff_{2^n}^\ast$ consider the polynomial $\zeta^k \mod f(\zeta)$ of degree $n-1$. The coefficients of this polynomial from an element of $\ff_2^n$. $\phi$ is just the inverse of this map. Also, assume that the all-zero vector is mapped to $0$ under $\phi$.  

Let $x\in \ff_2^n$ be the configuration to be hashed. Suppose the hash function is $h_{\nu,b}$, indexed by $\nu\in \ff_2^n$ and $b \in \ff_2^m$. The hash function is defined as follows:
%\begin{itemize}[noitemsep]
%\item %Randomly and uniformly pick 
Let  $\nu \in \ff_{2}^n$.
 Compute $z=\phi^{-1}(\phi(x) \cdot \phi(\nu) \mod f(\zeta)) \in \ff_2^n.$
 Let $y\in \ff_2^m$ be the first $m$ bits of  $z$.
  Finally, output $y+b$, where $b \in \ff_2^m$. % is a randomly and uniformly chosen vector.
%\end{itemize}

\begin{proposition}\label{prop:aff}
The hash function $h_{\nu,b}$ can be written as an affine transform ($x \mapsto Ax+b$) over $\ff_2^n$. 
\end{proposition}
\begin{proof} %[Proof of Prop.~\ref{prop:aff}]
It is sufficient to show that $z$ can be obtained as a linear transform of $\nu$. Note that the product of $\phi(x)$ and $\phi(\nu)$ can be written as a convolution between $x$ and $\nu\equiv (\nu_1, \nu_2, \dots, \nu_n)$ (as we can view this as product between two polynomials). Let $\Gamma$ be the $(2n -1) \times n$ matrix,
$$ \Gamma=
\begin{bmatrix}
\nu_1 & 0 & 0 &  \dots & 0\\
\nu_2 & \nu_1 & 0 & \dots & 0\\
\nu_3& \nu_2 & \nu_1 & \dots & 0\\
% \hdotsfor{5}\\
\vdots & \vdots & \vdots &\vdots & \vdots\\
\nu_n & \nu_{n-1} & \nu_{n-2} & \dots & \nu_1\\
0 & \nu_n & \nu_{n-1} & \dots & \nu_2\\
\vdots & \vdots & \vdots &\vdots & \vdots\\
0 & 0 & 0 & \dots & \nu_n
\end{bmatrix}.
$$
The reduction modulo $f(\zeta)$ can also be written as a linear operation. Just consider the $n \times (2n-1)$ matrix 
$P$ whose $i$th column contains the coefficients of the polynomial $\zeta^{i-1} \mod f(\zeta), 1\le i \le 2n-1.$ Note that the first $n$ columns of the matrix is simply the identity matrix.
We can write,
$
z= P\Gamma x.
$ 
\end{proof}
%The proof can be found in Appendix~\ref{app:der}.
%The following proposition is true by design.
%\begin{proposition}
Note that, to chose a random and uniform  hash function from $\{h_{\nu,b}, \nu \in \ff_2^n, b \in \ff_2^m\}$, one needs $m+n$ random bits.
%\end{proposition}
It follows that the  hash family is  pairwise independent. % of the which follows from known results.
%The main claim of this section now follows. This is a well-known fact in derandomization. 
%We provide the proof  in  Appendix~\ref{app:der} in the supplementary material  for completeness. 
\begin{proposition}\label{prop:pair}
The hash family $\{h_{\nu,b}, \nu \in \ff_2^n, b \in \ff_2^m\}$ is uniform and pairwise independent.
\end{proposition}
\begin{proof} %[Proof of Prop.~\ref{prop:pair}]
Suppose $\nu,b$ are randomly and uniformly chosen. For any $x_1,x_2\in \ff_2^n$ and $y_1,y_2 \in \ff_2^m$, first of all
$$
\Pr(h_{\nu,b}(x_1) = y_1) = \frac{1}{2^m},
$$ 
since $b$ is uniform. Now,
\begin{align*}
&\Pr(h_{\nu,b}(x_1) = y_1, h_{\nu,b}(x_2) = y_2) \\
&= \frac{1}{2^m}\Pr(h_{\nu,b}(x_2) = y_2 | h_{\nu,b}(x_1) = y_1)\\
&= \frac{1}{2^m}\Pr(h_{\nu,b}(x_2) -h_{\nu,b}(x_1)= y_2-y_1 | h_{\nu,b}(x_1) = y_1)\\
&= \frac{1}{2^m}\Pr(h_{\nu,b}(x_2) -h_{\nu,b}(x_1)= y_2-y_1).
\end{align*}
Now, since $\Pr((\phi(x_1)-\phi(x_2)) \cdot \phi(\nu) \mod f(\zeta) = u) = \frac1{2^n}$ for any $u$, we must have $\Pr(h_{\nu,b}(x_2) -h_{\nu,b}(x_1)= y_2-y_1)= \frac1{2^m}$.
Therefore the claim is proved.
%
%& = \frac{1}{2^m}\Pr((\phi(x_1)-\phi(x_2)) \cdot \phi(\nu) \mod f(\zeta) = u) \quad \text{ for a fixed }u\\
%& = \frac{1}{2^m} \cdot \frac{1}{2^m} = \frac{1}{2^{2m}}.
%\end{align}
\end{proof}

Moreover the randomness used to construct this hash function is also optimal. It can be shown that,  
%\begin{proposition}
the size of a pairwise independent hash family $\{h:\{0,1\}^n \to \{0,1\}^m\}$ is at least $2^{m+n}-2^{n}+1$  (see,  \cite{stinson1996connections}).
%\end{proposition}
This implies that $m+n$ random bits were essential for the construction.

\medskip

\noindent{\bf Construction 2: Toeplitz matrix.}
In \cite{ermon2013optimization}, a Toeplitz matrix was used as the hash function. In a Toeplitz matrix, each descending diagonal from left to right is fixed, i.e., if $A_{i,j}$ is the $(i,j)$th entry of a Toeplitz matrix, then $A_{i,j} = A_{i-1,j-1}$. So to specify an $m \times n$ Toeplitz matrix one needs to provide only $m+n-1$ entries (entries of the first row and first column).
Consider the  random $m \times n$ Toeplitz matrix $A_T$ where each of the entries of the first row and first column are chosen with equal probability from $\{0,1\}$, i.e., each entry in the first row and column is a Bernoulli($0.5$) random variable. The hash function $h_{A_T,b}: x\mapsto A_T x+ b$, is constructed by choosing a uniformly random $b\in \ff_2^m.$
%The proof of the following proposition is in  Appendix~\ref{app:der}.
\begin{proposition}\label{prop:top}
The hash family $\{h_{A_T, b}\}$  is uniform and pairwise independent \cite{ermon2013optimization}.
\end{proposition}
\begin{proof} %[Proof of Prop.~\ref{prop:top}]
%We can prove this claim by using arguments from \cite{mansour1990computational}. 
First of all, the uniformity of the family is immediate since $b$ is uniformly chosen. 
For any $x_1,x_2\in \ff_2^n$ and $y_1,y_2 \in \ff_2^m$,
$
\Pr(h_{A_T,b}(x_1) = y_1, h_{A_T,b}(x_2) = y_2) 
= \frac{1}{2^m}\Pr(h_{A_T,b}(x_2) = y_2 | h_{A_T,b}(x_1) = y_1)
%&= \frac{1}{2^m}\Pr(A_T(x_1-x_2)= y_1-y_2  | h_{A_T,b}(x_1) = y_1)\\
= \frac{1}{2^m}\Pr(A_T(x_1-x_2)= y_1-y_2).
$
It remains to prove that $\Pr(A_Tx =y) = \frac{1}{2^m}$ for any fixed $x,y$.  Let the $k$th coordinate of $x$ is the first to be in the support of $x$. %Then when we multiply $A_T$ with $x$, the $i$th 
Now consider the inner product of the $j$th row of $A_T$ with $x$. This product will contain the entry $A_T(j,k)$, the  $(j,k)$th entry of $A_T$. Note that, this entry would not have appeared in any of the inner products of $i$th row of $A_T$ and $x$, for $i <j$. Therefore the probability that this inner product is any fixed value is exactly $\frac12$ given inner product of all previous rows with $x$. Therefore, $\Pr(A_Tx =y) = \frac{1}{2^m}$. 
%Let $\supp(x) \subseteq \{1,2,\dots, n\}$ be the support of $x$.  
\end{proof}

Note that,  the number of random bits required from this construction is $2m+n-1$. %, in terms of memory requirement this hash family is slightly suboptimal to the previous construction. 
 Toeplitz matrix allow for much faster computation of the hash function (matrix-vector multiplication with Toeplitz matrix takes only $O(n \log n)$ time compared to $\Omega(mn)$ for unstructured matrices).
%\begin{theorem}
%\cite{mansour1990computational} Consider the hash family $\mathds{H}=\{h_{A,b}:\{0,1\}^{n} \rightarrow \{0,1\}^{m}\}$ where $h_{A,b}(x)= xA+b$ and $A$ is a
% toeplitz matrix of dimension $n \times m$ where the elements in the first row and column are sampled from Ber($0.5$) and $b$ is a random vector of dimension $m \times 1$. Now this hash function requires $2m+n-1$ random bits and this hash family is a $\epsilon=\frac{1}{2^{m}}$-SU hash family or a pairwise independent hash family.
%\end{theorem}
%\begin{proof}
%This can proved by simply showing that for any $\tau$ $\Pr (\tau A=0)=\frac{1}{2^{m}}$ . Consider the support of $\tau$ and take the first element where $\tau$ has a $1$ and choose the corresponding row of $A$.  Now it can be seen that every bit in that row acts as a Bernoulli($0.5$) noise for each column containing the bit. They are independent because if we consider 2 different columns, then the corresponding bits in that row are being used for the first time by $\tau$ (or else it will be a contradiction). Hence this proves our claim that the hash family is $\frac{1}{2^{m}} SU$. 
%\enPottsd{proof}

 We remark that {\em sparse Toeplitz Matrices} also can be used as our hash family, further reducing the randomness.
%\begin{remark}[Sparse Toeplitz Matrices]
In particular, we could  construct a Toeplitz matrix with Bernoulli($p$) entries for $p <0.5$.  While the pairwise independence of the hash family is lost, it is still possible to analyze the \texttt{MB-WISH} algorithm with this family of hashes since they form a {\em strongly universal} family \cite{stinson1996connections}.  The number of random bits used in this hash family is $(m+n-1)h(p)+m$.  This construction allows us to have sparse rows in the matrix for small values of $p$, which can lead to further speed-up. % but does not have the optimal value of $\epsilon$ since they are not pairwise independent. Let us now write the modified version of the WISH algorithm for our purpose.
%\end{remark}

Both the constructions of this section extend to $q$-ary alphabet straightforwardly.

%\section{Omitted Proofs from Section~\ref{sec:der}} \label{app:der}

\section{\texttt{MB-WISH} for computing permanent}\label{sec:perm}
For computing  the {\em permanent}, the domain of integration  is the symmetric group $S_n$. However $S_n$ can be embedded in  $\ff_q^n$ for a $q\ge n$. Therefore we can try to use \texttt{MB-WISH} algorithm and same set of hashes on elements of $S_n$ treating them as $q$-ary vectors, $q\ge n$. We need to be careful though since it is essential that the MAX-oracle returns a permutation  and not an arbitrary vector. The modified MAX-oracle for permanents therefore must have some additional constraints. However those being affine constraints, it turns out MAX-oracle is still implementable in  optimization softwares.

Recall the permanent of a matrix as defined in Eq.~\eqref{eq:perm}: ${\rm Perm}(D) \equiv \sum_{\sigma \in S_n} \prod_{i=1}^n D_{i,\sigma(i)}$. We will show that it is possible to  approximate the permanent with a modification of the \texttt{MB-WISH} algorithm and our idea of using multiple bins for optimization in the calls to MAX-oracle.
Also, recall from Section~\ref{sec:wish} that we set $\ff_q\equiv  \{\alpha_0, \alpha_1,\dots , \alpha_{q-1}\}$ where there exists a fixed ordering among the elements. We set $q \ge n$ and consider any $\sigma \in S_n$ as an $n$-length vector  over $\ff_q$ (that is by identifying $1, 2, \dots, n$ as $\alpha_0, \alpha_1, \dots, \alpha_{n-1}$ respectively). Then we define a modified hash family 
$\Hc_{m,n} = \{h_{A,b}:  A \in \ff_q^{m \times n}, b \in \ff_q^m \}$ with $h_{A,b}: S_n \to \ff_q^m: \sigma \mapsto A\sigma + b,$ the operations are over $\ff_q$.

However, when calling the MAX-oracle, we need to make sure that we are getting a permutation as the output. Hence the modified MAX-oracle for computing permanent will be:
\begin{align}
&\max_{\sigma \in \ff_q^n} w(\sigma)\nonumber\\ 
\text{ s.t., }
 A\sigma +b \prec \alpha_r\cdot \mathbf{1} ;&
 \sigma \prec \alpha_{n-1}\cdot \mathbf{1} ;
\sigma(i) \ne \sigma(j) \forall i \ne j,
\end{align}
where, $w(\sigma) =\prod_{i=1}^n D_{i,\sigma(i)}$.
These constraints ensures that the MAX-oracle returns a permutation over $n$ elements.
With this change we propose Algorithm~\ref{algo:permanent} to compute permanent of a matrix and 
call it \texttt{PERM-WISH}.  The full algorithm is provided as Algorithm \ref{algo:permanent}. % in Appendix~\ref{sec:perm_}.

%\section{The \texttt{PERM-WISH} Algorithm}\label{sec:perm_}

\begin{algorithm}[H]                   % enter the algorithm environment
\caption{\texttt{PERM-WISH} for and matrix $D$;  $\Omega=S_n$; weight function $w(\sigma) =\prod_{i=1}^n D_{i,\sigma(i)}$}                               \label{algo:permanent}
\begin{algorithmic}                 % enter the algorithmic environment
      \REQUIRE $\ell \rightarrow \lceil \frac{1}{\gamma}\ln \frac{2n}{\delta} \rceil$, $q>n$, $r=\lfloor{\frac{q-1}{2}}\rfloor, n'=\lceil n \log_{q/r} q\rceil$
      \STATE $M_0 \equiv \max_{\sigma \in S_n} w(\sigma)$ %$M_0 \equiv \max_{\sigma\in \ff_q^n} w(\sigma)$
       \FOR{$i\in\{1,2,\dots,n'\}$}
         \FOR{$k \in \{1,\dots,\ell\}$}
           \STATE Sample hash functions $h_i \equiv h_{A^{i},b^{i}}$ uniformly at random from $\Hc_{i,n}$ as defined in \eqref{eq:fam}% $h_{A,b}\{0,1,\dots,q-1\}^{n} \rightarrow \{0,1,\dots,q-1\}^{m}$
           \STATE $w_{i}^{(k)}=\max_{\sigma\in \ff_q^n} w(\sigma) $ such that $ A^{i}\sigma +b^i \prec \alpha_r\cdot \mathbf{1} ; \sigma \prec \alpha_{n-1}\cdot \mathbf{1} ;\sigma(k) \ne \sigma(l) \forall k \ne l$.
        \ENDFOR
         \STATE $M_{i}={\rm Median}(w_{i}^{(1)},w_{i}^{(2)},\dots,w_{i}^{(\ell)})$
       \ENDFOR
       \STATE Return $M_{0}+(\frac{q}{r}-1)\sum_{i=0}^{n'-1}M_{i+1}\big(\frac{q}{r}\big)^{i} $
\end{algorithmic}
\end{algorithm}

The main result of this section is the following.
\begin{theorem}\label{thm:perm}
Let $D$ be any $n\times n$ matrix. Let $q>n$ be a power of prime and  $r=\lfloor{\frac{q-1}{2}}\rfloor$.  For any  $\delta> 0$, % and a positive constant $\gamma = \frac{q}{3r}(\frac{1}{2}-\frac{r}{q})^{2}$, 
Algorithm \ref{algo:permanent} makes $\Theta(n^2 {\rm poly}(\log \frac{n}{\delta}))$ calls to the MAX-oracle and, with probability at least $1-\delta$ outputs a $(\frac{q}{r})^{2}=(4+O(1/n))$-approximation of $\mathrm{Perm}(D)$.
\end{theorem}

The proof of Theorem~\ref{thm:perm} follows the same trajectory as in Theorem~\ref{thm:mb}. The constraints in MAX-oracle ensures that a permutation is always returned. So in the proof of Theorem~\ref{thm:mb}, the $w_{i}^{(k)}$s can be though of as permutations instead in this setting. 
It should be noted that, we must take $q>n$ for \texttt{PERM-WISH} to work. That is the reason we get a $(4+O(1/n))$-approximation for the permanent.

It also has to be noted that, since $q$ is large, the straightforward extension of \texttt{WISH} algorithm would have provided only a $q^2 = n^2$-approximation of the permanent. Therefore the idea of using optimizations with multiple bins  are crucial here as it lead to a close to $4$-approximation.

\remove{\paragraph{Simulation.} 
We use the \texttt{PERM-WISH} algorithm to find the permanent of randomly generated matrices of size $n \times n, n=5,6,7,8,9,10$ and  compute the ratio of the value achieved by our method ($\hat{Z}$) and actual permanent ($Z$). For the purpose of constrained optimization, we use the python modules described above. For the experiment we take $q=11$ and $r=5$. We find $\hat{Z}/Z = 2.137, 1.103, 1.245, 2.407, 0.730, 1.527$ for $n =5,6,7,8,9,10$ respectively.}

%\begin{remark}[Constraints]
%It turns out that the constraints in the MAX-oracle in Algorithm \ref{algo:permanent} are linear/affine. Therefore they are still easy to implement in different CSP softwares.
%\end{remark}

\end{document}